\documentclass[runningheads,11pt]{llncs}
\usepackage[T1]{fontenc}

\usepackage{geometry}
\usepackage{amsmath,amssymb,bm,amsbsy,subfigure}
\usepackage{array,xspace,makeidx,verbatim,epsf,psfrag,calc,fancybox}
\usepackage{color,ifthen,graphicx,enumerate,float,url}
\usepackage{mathrsfs}
\usepackage{epstopdf}
\usepackage{nicefrac}
\usepackage{pstool}
\usepackage{indentfirst}
\usepackage{appendix}

\DeclareMathOperator*{\argmax}{arg\,max}

\renewcommand{\a}{\alpha}

\newcommand{\rw}[1]{{\color{red}#1}}

\newcommand{\dif}{\mathrm{d}}

\newcommand{\bmu}{\boldsymbol{\mu}}
\newcommand{\bx}{\bold{x}}
\newcommand{\bw}{\bold{w}}
\newcommand{\by}{\bold{y}}
\newcommand{\bz}{\bold{z}}
\newcommand{\btau}{\boldsymbol{\tau}}
\newcommand{\barcF}{\bar{\mathcal{F}}}
\newcommand{\cFbeta}{\mathcal{F}^\beta}
\newtheorem{assumption}{Assumption}

\begin{document}
\title{On the Effect of Time Preferences on the Price of Anarchy}
%
%
\author{Yunpeng Li\inst{1} \and
Antonis Dimakis\inst{2} \and
Costas A. Courcoubetis\inst{1}}
\authorrunning{Y. Li et al.}
%
\institute{The Chinese University of Hong Kong, Shenzhen, China\\
\email{\{liyunpeng,costas\}@cuhk.edu.cn} \and
Athens University of Economics and Business, Athens, Greece\\
\email{dimakis@aueb.gr}}

\maketitle
\begin{abstract}

This paper examines the impact of agents' myopic optimization on the efficiency of systems comprised by many selfish agents. In contrast to standard 
congestion games where agents interact in a one-shot fashion, in our model each agent chooses an {\em infinite sequence} of actions and maximizes the total reward stream {\em discounted over time} under different ways of computing present values. 
Our model assumes that actions consume common resources that get congested, and the action choice by an agent affects the completion times of actions chosen by other agents, which in turn affects the time rewards are accrued and their discounted value.
This is a {\em mean-field game}, where an agent's reward depends on the decisions of the other agents through the resulting action completion times. For this type of game we define stationary equilibria, and analyze their existence and {\em price of anarchy} (PoA). 

Overall, we find that the PoA depends entirely on the type of discounting rather than its specific parameters. 
For {\em exponential discounting}, myopic behaviour leads to extreme inefficiency: the PoA is $\infty$ for any value of the discount parameter. 
For {\em power law discounting}, such inefficiency is greatly reduced and the PoA is 2 whenever stationary equilibria exist. This matches the PoA when there is no discounting and players maximize long-run average rewards.  
Additionally, we observe that exponential discounting may introduce unstable equilibria in learning algorithms, if action completion times are interdependent.
In contrast, under no discounting all equilibria are stable.
\end{abstract}


\section{Introduction}

In crowdsourcing, access, and sharing economies, a large number of individuals is repeatedly interacting to exchange goods and services, with each individual pursuing his or her own interest. 
At least as early as in \cite{Pigou}, it  has been observed that selfish behavior may lead to inefficient outcomes (from the point of view of the system as a whole), e.g., see \cite{agt} and references therein, where the {\em price of anarchy (PoA)} is used to quantify the degree of inefficiency.
While the efficiency of a system is usually measured in the long term, e.g., by the profit over a long time period, individual agents prefer earlier gratification than one later on.
An aspect that, to the best of our knowledge, has not been considered before in the PoA literature, is
the effect the {\em time preference}, i.e., way of discounting future revenue, of the individual agents have on the PoA.

The price of anarchy has been  studied extensively for selfish routing games in networks, where sharp results exist \cite{roughgarden}. Since agents interact in a limited `one-shot' way, the concept of time preferences is not relevant in this type of games. Similarly, in games between dynamic flows (also known as flows over time) \cite{correa,skutella}, where the notion of time is explicit, the focus is on performance measures such as throughput where discounting is again not meaningful.

In contrast, exponential discounting of rewards has been considered in {\em mean-field games (MFG)}, e.g., in \cite{adlakha,hopenhayn,huang,jovanovic}, where repeated interactions between a continuum of agents of a constant mass take place. The existing literature is primarily focused on existence of equilibria and only a few works tackle PoA \cite{tomlin,carmona_poa,cd23}, but without discounting.

We consider a finite action mean field  game  between a population of agents of constant mass. Each agent chooses actions from a finite set, where each
action has its own reward and completion time. In contrast to standard 
congestion games where agents interact in a one-shot fashion, in our model each agent chooses an {\em infinite sequence} of actions and maximizes the total reward stream {\em discounted over time} under different ways of computing present values for rewards obtained further in the future.
A basic feature of our model is that the completion time of an action chosen by an agent depends in a general way on the number of agents that execute other actions. This is a case of systems where actions consume common resources that get congested, and thus the action choice by an agent affects the completion times of actions chosen by other agents and the amount their rewards get discounted.

In contrast to  finite state finite action mean field games \cite{guo2022},  the game we consider is stateless, i.e., can be represented using a single state, 
and is classified in the literature as a `stateless mean field game' \cite{Yardim2023}, also related to multi-armed bandit games \cite{Gummadi2013,WangJ21}. Unlike stateful mean field games, our model does not incorporate state transitions. When we consider the average reward criterion, our model can be reduced to the one investigated in \cite{Yardim2023}. However, when considering a more realistic discounted reward criterion, our problem exhibits a unique feature: the population distribution over actions affects both the frequency of receiving rewards {\em and} the amount of discounting these actions receive.
These aspects, although fundamental to {\em understand myopic behaviour in the case of congestion}, have not been explored so far in the literature.

 Although it may seem simple,  stateless mean field games have a range of applications in areas such as computer networks, wireless communications, and transportation, where the communication or travel delay of a link is influenced by the number of users accessing it. Specifically,
our MFG formulation is inspired by the following examples, with the first two adapted from \cite{Banez2021,Yardim2023}.

\textbf{Resource  allocation in cloud-edge systems.}
Consider a distributed computing architecture with \(n\) virtualized resource pools (e.g., edge computing nodes) serving  a large population of concurrent users. These nodes exhibit heterogeneous processing capacities and load-dependent latency penalties, compounded by virtualization-induced interference effects. Each user dynamically selects a node for task offloading, experiencing congestion-induced latency that scales superlinearly with the instantaneous user density on the chosen node. The resultant latency functions defy explicit modeling due to complex interdependencies in power allocation, thermal constraints, and memory bandwidth sharing across virtual machine.  This complexity necessitates an MFG-based analysis to explore these intricate dynamics effectively.

\textbf{Spectrum sharing in 5G-UAV networks.} Multi-agent spectrum access challenges in 5G-UAV networks go beyond traditional collision-avoidance frameworks. When a large number of unmanned aerial vehicles (UAVs) compete for 
\(n\) licensed spectrum bands, simultaneous transmissions on the same frequency result in probabilistic interference rather than complete signal loss. This behavior resembles time-division protocols in 5G-NR networks, where overlapping transmissions lead to throughput degradation governed by stochastic backoff mechanisms. MFG can effectively model such `soft collision' dynamics, facilitating distributed coordination among agents without the need for explicit message passing.

\textbf{Spatial competition among food delivery couriers.} In urban food delivery systems, a large population of couriers dynamically select service zones (e.g., shopping mall hotspots) based on real-time order frequency and competitor distribution. Each courier's revenue depends nonlinearly on the spatial order density: high-density zones offer a high frequency of orders but intensify competition, while low-density areas reduce contention at the cost of increased idle time and delivery distances. The absence of centralized coordination necessitates MFG-based analysis where agents' spatial distributions are  formed through strategical interactions. 
 
We distill the above examples into a simplified toy model, presented in Appendix \ref{Motivating Example}.
The existing formulations of finite state finite action  mean field games, some of which also addresses scenarios similar to those mentioned above, are typically modeled in discrete time (e.g., \cite{guo2022,Yardim2023}), with delay or latency incorporated into the reward structure. In contrast, we propose a continuous-time formulation of mean field games. This approach naturally aligns with the applications described and is exceptionally rare in the literature, with the only comparable example being \cite{neumann}. Our novel yet intuitive formulation of mean field games highlights the necessity and relevance of studying the impact of time preference on game efficiency.

We consider the effect of two types of devaluation over time: {\em exponential} and {\em power-law discounting}.
The former is a widely used model owing to its tractability due to representing {\em consistent} time preferences: the preference between rewards received at different times does not change under equal time shifts. Power-law discounting allows inconsistent preferences, and in certain contexts it is a more plausible model.  (E.g., hyperbolic discounting \cite{Prelec2004} --a particular type of power law discounting-- is supported by evidence from neuroscience \cite{HAMPTON2017336}.)

Each agent is assumed to maximize the total present value of all future rewards accrued to her, and we assess the average rate of rewards generated over time by all agents in a system equilibrium (if such an equilibrium exists; the equilibirum analysis of this system, although simple to define, is in general highly non-trivial). We compare this against the maximum reward rate obtainable if the agents were assigned actions optimally by a central planner. 

Our main findings are:
\begin{itemize}
 \item Under exponential discounting, the mean field game is equivalent to a {\em stable population game} \cite{Sandholm2009}, and an equilibrium always exists. The PoA is $\infty$ for any value of the discount factor. An explicit expression of PoA is provided, in terms of appropriately defined indices for each action. This is useful in obtaining the PoA in restricted classes of instances but more importantly, it reveals how the parameters determine the PoA. 


\item Under power law discounting, stationary equilibria do not always exist. When the discounting parameter $\alpha>1$, there may be no stationary best-response to any stationary distribution of agents and hence no stationary equilibrium. But, if such an equilibrium exists, the PoA is $2$. When $\alpha\leq1$, a stationary equilibrium always exists. In this case, the mean field game is equivalent to a stable population game (as in the case of exponential discounting) and the  PoA is $2$.

\item Exponential discounting may introduce unstable equilibria in learning dynamics. We show this by analyzing a system with two actions but with more intricate dependencies between their completion times. In contrast, we show that under no discounting all equilibria are stable, as also observed to hold, e.g., for the replicator dynamics in selfish routing \cite{Fischer2004}.

\end{itemize}

The paper is organized as follows.
Section 2 defines the class of mean-field games under consideration and the equilibrium concepts that we study. The case of exponential discounting is examined in Section 3, where we examine the existence  and stability of a stationary equilibrium, and calculate the PoA. In Section 4, we deal with power law discounting. We discuss when stationary equilibria exist and derive the PoA when such equilibria exist. In Section 5, we discuss how discounting can add more equilibria and affect stability. We conclude with a brief discussion on some open questions in Section 6.

\section{Game Model}
\label{sec:model}

We consider a mass $m$ of nonatomic agents where each agent chooses from a set $A=\{1,2,\ldots,n\}$ of \emph{actions}, with $n\ge2$. 
Action $i\in A$ takes $\tau_i$ time units to complete and {\em after that time} a reward $r_i$ is accrued to the acting agent. 
We call $r_i$ and $\tau_i$ the \emph{reward} and \emph{sojourn time} associated with action $i$, respectively, and let $\boldsymbol{r}=(r_i)_{i\in A}$ and $\boldsymbol{\tau}=(\tau_i)_{i\in A}$.
After an action completes, agents choose their next action and continue to do so over an infinite time horizon. 
A key feature of our model is that the sojourn times are not fixed and depend on the fraction of the population that choose the different actions. The form of this dependence is described in the following subsection. 


Let ${\mathcal P}(A)$ be the set of probability measures over the set of actions $A$, i.e., all $n$-dimensional probability vectors. Then, $\sigma_k \in {\mathcal P}(A)$ is the strategy of the agent to choose its $k$th action, i.e.,
$\sigma_k(i)$ is the probability that the agent chooses $i\in A$ as his $k$th action.
A \emph{strategy} $\boldsymbol{\sigma}$ is defined as an infinite sequence of probability measures, $\boldsymbol{\sigma}=(\sigma_k)_{k\in \mathbb{N}}\in \mathcal{P}(A)^{\mathbb{N}}$. A strategy $\boldsymbol{\sigma}$ is \emph{stationary} if $\sigma_k$ is the same for all $k$; we use $\mathcal{P}(A)$ to denote the set of stationary strategies.
A stationary strategy $\sigma\in {\mathcal P}(A)$ is \emph{deterministic} if $\sigma(i)\in \{0, 1\}$ for all $i\in A$. 
With a slight abuse of notation, we denote by $i$ a deterministic stationary strategy with $\sigma(i) = 1$ and $\sigma(j) = 0, i\neq j$.


%
We assume that the utility (net present value) of receiving an amount of reward $r$ at time $t$ is $u(r, t)$, and that all agents use the same utility function.
For exponential discounting, $u(r,t)=re^{-\beta t}$ where $\beta$ is a positive constant, and
$u(r,t)=rt^{-\a}$ with $\alpha>0$ for power law discounting.  
Higher values of 
$\alpha, \beta$ indicate more myopic agents, focusing more on their immediate payoff rather than the future. 

We consider a single agent out of the mass of the remaining agents. His objective is to maximize the sum of utilities for each reward received over an infinite time horizon. 
 We assume that the system is stationary, i.e., the rest of the agents use a stationary strategy to choose their actions. This implies some stationary value $\boldsymbol{\tau}$ for the action sojourn times, which is the one observed by the given agent when he makes decisions.
 
 Given the agent's strategy $\boldsymbol{\sigma}=(\sigma_k)_{k\in \mathbb{N}}$, let $a_j$ be the $j$th action the agent chooses. Then, at time $\sum_{j=1}^k\tau_{a_j}$, when the $k$th action completes, the agent receives reward $r_{a_k}$.
 The {\em (expected) total utility} under $\boldsymbol{\sigma}$  is
 \begin{equation}\label{total_reward_individual_agent}
V(\boldsymbol{\sigma},\boldsymbol{\tau})=\mathbb{E}_{\boldsymbol{\sigma}}\left[\sum_{k=1}^\infty u\bigg(r_{a_k},\sum_{j=1}^k\tau_{a_j}\bigg)\right],
 \end{equation}
given the sojourn times $\boldsymbol{\tau}$, where the expectation is taken with respect to $\boldsymbol{\sigma}$.

 A strategy $\boldsymbol{\sigma}^*$ is optimal if 
 \begin{equation}\label{condition_optimal_strategy}
     V(\boldsymbol{\sigma}^*,\boldsymbol{\tau})=\sup\limits_{\boldsymbol{\sigma}\in \mathcal{P}(A)^{\mathbb{N}}}
     V(\boldsymbol{\sigma}, \boldsymbol{\tau}).
 \end{equation}

Note that an agent's total utility $V(\boldsymbol{\sigma}, \boldsymbol{\tau})$ is determined by his own strategy $\boldsymbol{\sigma}$ as well as those of the other agents
through the effect they have on the vector of sojourn times $\boldsymbol{\tau}$. 
This is a mean field game\footnote{A mean field game (MFG) is a game-theoretical framework that models the strategic decision-making of a large number of interacting agents \cite{jovanovic,huang,hopenhayn,adlakha}. In MFGs, each agent interacts with the  `mean field', which represents the average effect of the strategies taken by  all the other agents. In our model, the mean field corresponds to the distribution over the set of actions (defined in the following section), which directly determines the values of sojourn time $\tau$.}, where the interactions among a large number of agents can be approximated by their average effect on each agent, i.e., the {\em mean-field value}, which in our case is $\boldsymbol{\tau}$.


In summary, the game we consider is of a mean field type with a mass $m$ of nonatomic agents as players, $\mathcal{P}(A)^{\mathbb{N}}$ as the set of strategies, $V(\boldsymbol{\sigma},\boldsymbol{\tau})$ in \eqref{total_reward_individual_agent} as payoffs, and
$\boldsymbol{\tau}$ as the mean field term.  We denote the stateless mean field game investigated in this work by $\mathcal{G}=(m,A,\boldsymbol{r},\boldsymbol{\tau},u)$. In this paper we make no assumptions that the reader is familiar with  MFG theory, and our analysis is self-contained. 

\subsection{Sojourn Times}



For clarity and simplicity in both presentation and analysis, we assume that sojourn times are deterministic for a given agent distribution over actions. This assumption is sensible when the system operates in a steady state. A more realistic and general approach would involve sojourn times following a general probability distribution that evolves dynamically based on the agent distribution over actions. However, deterministic sojourn times can be interpreted as a one-point mass distribution, representing an extreme point in the set of all possible sojourn time distributions.
This assumption aligns with our focus on studying the price of anarchy, which examines worst-case scenarios. In fact, under exponential discounting, the one-point mass distribution indeed represents the worst-case scenario, a result that can be demonstrated using Jensen's Inequality. In future work, we plan to extend our analysis to accommodate general sojourn time distributions, allowing for a broader and more realistic exploration of the system dynamics. In this work, deterministic sojourn times serve our purpose of investigating the effects of time preference issues.

We assume stationarity in how agents make decisions, and sojourn times are determined by the \emph{stationary} distribution $\boldsymbol{\mu} = (\mu_i)_{i\in A}$ of agents over the actions, i.e., for each action $i\in A$, the mass $\mu_i$ of agents who continuously choose $i$. Note that since the mass of the agents is infinitely divisible, this definition is equivalent with agents choosing actions with fixed probabilities.
For any total agent mass $m>0$, let $\mathcal{U}_m(A)=\{\boldsymbol{\mu}\in\mathbb{R}_+^{n}\mid \sum_{i\in A}\mu_i$$=m\}$, i.e, the set of all possible mass distributions over $A$.
Thus, we treat $\tau_i$ as functions from $\mathcal{U}_m(A)$ to ${\mathbb{R}_+}$,
with
$\boldsymbol{\tau}(\boldsymbol{\mu}) = (\tau_i(\boldsymbol{\mu}))_{i\in A}$.
 Using $\tau_i$, we define the {\em action rate} function 
\begin{equation}
    \label{eq:little_law}
x_i(\boldsymbol{\mu}) = 
    \frac{\mu_i}{\tau_i(\boldsymbol{\mu})},
\end{equation}
on $\mathcal{U}_m(A)$, and let $\boldsymbol{x}(\boldsymbol{\mu}) = (x_i(\boldsymbol{\mu}))_{i\in A}$. 
This condition expresses Little's law: at any time, the average mass of agents in a service facility where the action is executed, is equal to the average rate of agents entering the facility  multiplied by the average time an agent spends in the facility. 
Thus,
$x_i(\boldsymbol{\mu})$ is interpreted as the rate (in mass per unit of time) of agents executing action $i\in A$.

We assume $\boldsymbol{\tau}$ satisfies the following monotonicity property that relates sojourn times to action rates\footnote{Assumption \ref{assum_sojourntime} stems from closed queuing network analysis \cite{kelly1989} and is also a common assumption in the congestion games literature (e.g., see \cite{roughgarden}): as a resource is used at a higher rate the congestion delay is nondecreasing.}:
\begin{assumption}
\label{assum_sojourntime}
The function $\boldsymbol{\tau}: \mathcal{U}_m(A)\rightarrow \mathbb{R}_+^{n}$
is continuous and rate monotonous, i.e.,
     for any $\boldsymbol{\mu},\boldsymbol{\mu}^\prime\in\cup_{m>0}\, \mathcal{U}_m(A)$ s.t. $x_i(\boldsymbol{\mu})< x_i(\boldsymbol{\mu}^\prime)$ for some $i$, then $\tau_{i}(\boldsymbol{\mu})\leq \tau_i(\boldsymbol{\mu}^\prime)$.
\end{assumption}

To make the connection of Assumption \ref{assum_sojourntime} with congestible resources, we provide the next example of sojourn times arising in a fluid queueing model. \\

{\bf A Fluid Sojourn Time Model}. 
We consider a resource model with independent actions and `parallel' resources.  More specifically, we assume each action $i$ consumes one unit of resource $l_i$ per unit time, and the  replenishment rate of resource $l_i$ is $b_i$.  The resources $l_i$'s,  $i\in A$, are `in parallel', i.e., independent of each other (if resources are not parallel, actions may require subsets of  resources that intersect). 
Assume for each action $i$, there is an action execution time for this action to be completed,  which is an increasing function of the rate $x_i$, denoted by $t_i(x_i)$.
Given any mass distribution $\bmu$, because of the resource constraint we need $x_i(\bmu)\le b_i$, and $x_i(\bmu)$ is the consumption rate of resource $l_i$.
If $x_i(\bmu)<b_i$, the sojourn time is simply $\tau_i(\bmu)=t_i(x_i(\bmu))$  since there is no waiting for the resource. If $x_i(\bmu)=b_i$, the sojourn time is  $\tau_i(\bmu)=t_i(x_i(\bmu))+w_i(\bmu)$  where 
$w_i$ is the time agents need to wait for the resource $i$ to become available before they use it. To compute the sojourn time vector $\btau(\bmu)$, we need to solve a system of equations  with unknowns $x_i$'s and $w_i$'s using Little's law \eqref{eq:little_law} and the resource constraint. In fact,  $\btau(\bmu)$ is uniquely determined by the mass distribution $\bmu$ as one can easily check that $(x_i(\boldsymbol{\mu}))_{i\in A}$ is the optimal solution
of the convex optimization problem:
\begin{align}
\max_{\bx\in\mathbb{R}^{n}_+}& \quad \sum_{i\in A} {\mu}_{i} \log(x_{i})   - \sum_{i\in A}\int_0^{x_i}t_i(y)\dif y \label{convex_optimization} \\
\text{s.t.} &\quad   x_i\le b_i,\quad\forall\>i\in A.\label{constraint:resource_i}
\end{align}
Since the objective function is strictly concave, the optimal solution is unique. Furthermore, the optimality conditions imply
that the $w_i(\boldsymbol{\mu})$'s are the optimal Lagrange multipliers for the resource constraints  \eqref{constraint:resource_i}. A 
more detailed description of this fluid queueing model and a formal proof can be found in Appendix \ref{appendix:fluid_model}. Note that in this resource model, $x_i(\boldsymbol{\mu}) < x_i(\boldsymbol{\mu}')$ implies $x_i(\boldsymbol{\mu}) < b_i$ so $\tau_i(\boldsymbol{\mu}) = t_i(x_i(\bmu))\leq t_i(x_i(\bmu')) \leq \tau_i(\boldsymbol{\mu}')$. Thus, Assumption~\ref{assum_sojourntime} is satisfied.

In the results which follow, the sojourn times are only assumed to satisfy Assumption \ref{assum_sojourntime}, which generalizes the `parallel' resources of the fluid queueing model introduced above.
Indeed, the upper bounds for the PoA rest entirely on Assumption \ref{assum_sojourntime}; no underlying resource model is necessary. 
Nevertheless, the fluid queueing model is convenient for providing lower bounds.
The special case where every action has a constant execution time, i.e., the $t_i$'s are constant, is particularly useful for this purpose; we refer to this as the \emph{constant action execution time model}.

\subsection{Stationary Equilibrium}
The following definition is similar to that of a stationary mean-field equilibrium in \cite{lions,adlakha}, but more appropriate in our context.

  \begin{definition}\label{def:stationary_ne}
 A {\em stationary equilibrium} of the game $\mathcal{G}$ is a pair of a mass distribution $\boldsymbol{\mu}^\dagger\in\mathcal{U}_m(A)$ and a stationary strategy $\boldsymbol{\sigma}^\dagger\in\mathcal{P}(A)$ 
such that
\begin{enumerate}
    \item[(a)] 
   $\boldsymbol{\mu}^\dagger$ is compatible with the stationary strategy $\boldsymbol{\sigma}^\dagger$, i.e., 
   $\mu^\dagger_i = \theta \sigma^\dagger(i) \tau_i(\boldsymbol{\mu}^\dagger)$ for some  normalizing $\theta$ such that $\sum_{j\in A}\mu^\dagger_j=m$.
    \item[(b)] $\boldsymbol{\sigma}^\dagger$ is optimal among all possible strategies, i.e.,
    \begin{equation}
     \label{eq:equilibrium_condition}
     V(\boldsymbol{\sigma}^\dagger,\boldsymbol{\tau}(\boldsymbol{\mu}^\dagger))=\sup\limits_{\boldsymbol{\sigma}\in \mathcal{P}(A)^{\mathbb{N}}}V(\boldsymbol{\sigma},\boldsymbol{\tau}(\boldsymbol{\mu}^\dagger)).
 \end{equation}
\end{enumerate}

 \end{definition}

If $(\boldsymbol{\sigma}^\dagger,\boldsymbol{\mu}^\dagger)$ is a stationary equilibrium, then all agents use the same stationary strategy $\boldsymbol{\sigma}^\dagger$ resulting in a  mass of agents proportional to $\sigma^\dagger(i) \tau_i(\boldsymbol{\mu}^\dagger)$ choosing action $i$ where the scale coefficient is $ \theta=\sum_{j\in A}\frac{\mu^\dagger_j}{\tau_j(\boldsymbol{\mu}^\dagger)}$.
Furthermore, given the mass distribution $\boldsymbol{\mu}^\dagger$ and the resulting sojourn times $\boldsymbol{\tau}(\boldsymbol{\mu}^\dagger)$, no agent wants to deviate from the strategy $\boldsymbol{\sigma}^\dagger$, i.e.,
 \eqref{eq:equilibrium_condition} holds. Note that  it suffices to specify a stationary equilibrium by the equilibrium distribution $\boldsymbol{\mu}^\dagger$ as the stationary strategy $\sigma^\dagger$ can be computed as a function of $\boldsymbol{\mu}^\dagger$ according to Definition \ref{def:stationary_ne}.(a) and the above formula of $\theta$. 
We denote the set of stationary equilibria of $\mathcal{G}$ by $\mathcal{S}(\mathcal{G})$ which is a subset of $\mathcal{P}(A)$.
 


In Section \ref{section:exponential_discounting}, we show that a stationary equilibrium always exists under exponential discounting. Although it does not always exist under power law discounting,  we will provide conditions where it does exist in Section \ref{section:nonexponential_discounting}.

\subsection{Social Optimum and Price of Anarchy}
We define the {\em social welfare} as the long-term  average (undiscounted) reward per unit time generated by the $m$ mass of agents.
Given a mass distribution $\boldsymbol{\mu}\in\mathcal{U}_m(A)$, the social welfare is given by 
\begin{equation}
    \label{eq:social_welfare}
    \text{SW}(\boldsymbol{\mu})=\sum_{i=1}^{n}r_ix_i(\boldsymbol{\mu}).
\end{equation}
A distribution $\boldsymbol{\mu}^*$ is \emph{optimal} if
\begin{equation}
    \label{eq:social_welfare_optimal}
    \boldsymbol{\mu}^*\in\mathcal{O}(\mathcal{G})=\argmax_{\boldsymbol{\mu}\in\cup_{m^\prime\le m}\mathcal{U}_{m^\prime}(A)}\text{SW}(\boldsymbol{\mu}),
\end{equation}
where we denote the set of optimal distributions by $\mathcal{O}(\mathcal{G})\subseteq\cup_{m^\prime\le m}$
$\mathcal{U}_{m^\prime}(A)$. Note that we do not require $\boldsymbol{\mu}^*\in\mathcal{U}_m(A)$ as sometimes adding mass may increase congestion and in fact decrease the average reward rate.

The PoA is defined as the largest possible ratio between  the optimal social welfare $\text{SW}(\boldsymbol{\mu}^*)$ and the social welfare $\text{SW}(\boldsymbol{\mu}^\dagger)$ of a stationary equilibrium $\boldsymbol{\mu}^\dagger$, i.e.,
\begin{equation}
    \label{eq:PoA_definition}
    \text{PoA}= \sup_{{\substack{\{(\mathcal{G}, \boldsymbol{\mu}^*, \boldsymbol{\mu}^\dagger) | 
    \boldsymbol{\mu}^* \in \mathcal{O}(\mathcal{G}),
\boldsymbol{\mu}^\dagger\in\mathcal{S}(\mathcal{G}) \}}}} 
    \frac{\text{SW}(\boldsymbol{\mu}^*)}{\text{SW}(\boldsymbol{\mu}^\dagger)}\,,
\end{equation}
where the supremum is taken over all instances of game  $\mathcal{G}$
and corresponding stationary equilibria $\boldsymbol{\mu}^\dagger$, and optimal distributions $\boldsymbol{\mu}^*$. Note that social welfare itself does not involve discounting, as it involves the aggregate benefit.

\section{Exponential Discounting}\label{section:exponential_discounting}
In this section the reward each agent receives is exponentially discounted over time, i.e., $u(r,t)=re^{-\beta t}$. 

\subsection{Existence of Stationary Equilibrium }\label{subsection:se_exponential_discounting}
In a stationary equilibrium, any action which is taken by a positive mass of agents must generate the greatest total utility.
The existence of a stationary equilibrium is established in the following theorem.

\begin{theorem}\label{thm:stationary_equilibrium}
There exists a stationary equilibrium for the game $\mathcal{G}$ with exponential discounting. Furthermore, $\boldsymbol{\mu}^\dagger$ is a stationary equilibrium of $\mathcal{G}$ if and only if 
\begin{equation}
\label{thm:stationary_equilibrium_eq}
    \frac{r_i}{e^{\beta \tau_i(\boldsymbol{\mu}^\dagger)}-1}\ge \frac{r_j}{e^{\beta \tau_j(\boldsymbol{\mu}^\dagger)}-1},\quad\forall\>i\in A^\dagger,\>j\in A,
\end{equation}
where $A^\dagger=\{i\in A\mid \mu^\dagger_i>0\}$.
\end{theorem}


\emph{Sketch of Proof}: Because of exponential discounting, given any mass distribution $\mu$ over the action space, there exists a stationary strategy which is a best response for the representative agent. This is also known as Bellman's principle of optimality. Now by representing mean field equilibria as fixed points of the best response correspondence, we can apply Kakutani's fixed point theorem to show the existence of stationary equilibrium.
The detailed proof is given in Appendix \ref{appendix:proof_thm:stationary_equilibrium}.\qed

In certain cases, the characterization of equilibrium in Theorem \ref{thm:stationary_equilibrium} is sufficient to completely determine the equilibrium, e.g., this is done in Appendix \ref{appendix:se_example_1} for the constant action execution time model.

\subsection{Stability}
In view of Theorem \ref{thm:stationary_equilibrium}, the stateless mean field game $\mathcal{G}$ is in fact equivalent to a \emph{population game}. Moreover, if Assumption \ref{assum_sojourntime} on the sojourn time vector $\bold{\tau}$ holds, $\mathcal{G}$ is a \emph{stable population game}. 
\footnote{A population game is a type of nonatomic game where  agents aims to optimize their \emph{immediate payoffs} over a finite set of strategies. Population games are of particular interest because various well-known evolutionary dynamics—such as best response, replicator, and projection dynamics—have been developed to compute equilibria and explain equilibrium formation \cite{Sandholm2010}. If a population game is \emph{stable}, various evolutionary dynamics are guaranteed to converge to equilibria (Sandholm, 2009).
Population game can be defined by a payoff vector, with each element representing the payoff associated with a specific strategy. These payoffs are determined as functions of the distribution over the set of strategies.}
Consider a population game $\mathcal{F}$ with $m$ nonatomic players, set of strategies $A$, and the payoff functions given by,
\[\boldsymbol{F}(\boldsymbol{\mu})=\left(\frac{r_i}{e^{\beta\tau_i(\boldsymbol{\mu})}-1}\right)_{i\in A},\]
which coincide with the total discounted reward vector of the deterministic stationary strategies in $\mathcal{G}$. 
By the definition of  stable game  \cite{Sandholm2009}, $\mathcal{F}$ is stable if
\begin{equation}
    \sum_{i=1}^{|A|}\Big(\frac{r_i}{e^{\beta\tau_i(\boldsymbol{\mu})}-1}-\frac{ r_i}{e^{\beta\tau_i(\boldsymbol{\mu}^\prime)}-1}\Big)(\mu_i-\mu_i^\prime)\le0.\label{condition:stable}
\end{equation}
It turns out that rate monotonicity of $\boldsymbol{\tau}$ (Assumption \ref{assum_sojourntime}) is a sufficient condition for $\mathcal{F}$ being stable.
\begin{proposition}\label{prop:stable_game}
    The stateless mean field game $\mathcal{G}$ with exponential discounting is equivalent to a population game $\mathcal{F}$, i.e., the set of stationary equilibria of $\mathcal{G}$ coincide with the set of Nash equilibria of $\mathcal{F}$. If Assumption \ref{assum_sojourntime}  holds, $\mathcal{F}$ is a stable game.
\end{proposition}

\emph{Sketch of Proof}: 
 The equilibria of $\mathcal{G}$ belongs to  $\mathcal{P}(A)$ and the same is true for $\mathcal{F}$. 
 By Theorem \ref{thm:stationary_equilibrium}, $\mathcal{F}$ and $\mathcal{G}$ have exactly the same equilibrium condition. Hence,
 $\mathcal{G}$ is equivalent to $\mathcal{F}$.
 Now we only need to show that if Assumption \ref{assum_sojourntime}  holds, $\mathcal{F}$ is stable, i.e., \eqref{condition:stable} holds. 
 Note that a sufficient condition for  \eqref{condition:stable} is that for any $\boldsymbol{\mu}$ and $\boldsymbol{\mu}^\prime$ and any $i\in A$,
 \[\mu_i\le\mu_i^\prime\implies\tau_{i}(\boldsymbol{\mu})\le\tau_i(\boldsymbol{\mu}^\prime).\]
 We prove that Assumption \ref{assum_sojourntime} implies this sufficient condition by contradiction.
Assume for some $\boldsymbol{\mu}$ and $\boldsymbol{\mu}^\prime$, $\mu_i\le\mu_i^\prime$, but $\tau_{i}(\boldsymbol{\mu})>\tau_i(\boldsymbol{\mu}^\prime)>0$, then $x_i(\boldsymbol{\mu})<x_i(\boldsymbol{\mu}^\prime)$, and hence by Assumption \ref{assum_sojourntime}, $\tau_{i}(\bmu)\leq \tau_i(\bmu^\prime)$, contradiction! \qed

Since $\mathcal{F}$ is a stable game, the set of Nash equilibria of $\mathcal{F}$ is convex. Moreover, it is globally asymptotically stable under many interesting dynamics like best response dynamic. For more details, we refer the readers to \cite{Sandholm2009}. The stability of  $\mathcal{F}$ relies on Assumption \ref{assum_sojourntime}. In Section \ref{section:stability}, we will see that for a more general sojourn time model violating Assumption \ref{assum_sojourntime}, the equilibrium could even be not locally asymptotically stable.

\subsection{Price of Anarchy}\label{subsection:poa_exponential_discounting}
For any instance of game $\mathcal{G}$ with exponential discouting, we define 
\begin{equation}\label{parameter_mi}
    \chi_i\big(\mathcal{G}\big)=\chi_i(m,A,\boldsymbol{r},\boldsymbol{\tau},u)=\sup_{\boldsymbol{\mu}\in\mathcal{U}_m(A)}\frac{e^{\beta \tau_i(\boldsymbol{\mu})}-1}{\beta \tau_i(\boldsymbol{\mu})},\quad\forall\,i\in A\,.
\end{equation}
Note that as $f(x)=\frac{e^x-1}{x}$ is an increasing function of $x$ for $x>0$, $\chi_i(\mathcal{G})$ is increasing in $\beta$ and $\tau_i$.
We define a characteristic number of $\mathcal{G}$ by
\begin{equation}\label{parameter_m}
    \chi(\mathcal{G})=\max_{i\in A} \chi_i(\mathcal{G}).
\end{equation}
Note that $\chi(\mathcal{G})$ is a function of game parameters $m$, $A$, $\boldsymbol{\tau}$, and $\beta$, but independent of the reward vector $\boldsymbol{r}$.
The PoA of $\mathcal{G}$ is determined by the 
characteristic number  $\chi(\mathcal{G})$. 
\begin{proposition}\label{prop:poa_m+1}
 Given Assumption \ref{assum_sojourntime} on the sojourn time vector $\boldsymbol{\tau}$,
the PoA is $M+1$ when confined to set of instances of  $\mathcal{G}$ with characteristic number $\chi(\mathcal{G}) = M$, i.e.,
\begin{equation*}
    \sup_{{\substack{\{(\mathcal{G}, \boldsymbol{\mu}^*, \boldsymbol{\mu}^\dagger) \mid\chi(\mathcal{G}) = M, 
    \boldsymbol{\mu}^* \in \mathcal{O}(\mathcal{G}),
\boldsymbol{\mu}^\dagger\in\mathcal{S}(\mathcal{G}) \}}}} \frac{\text{SW}(\boldsymbol{\mu}^*)}{\text{SW}(\boldsymbol{\mu}^\dagger)} = M + 1.
\end{equation*}
\end{proposition}

\emph{Sketch of Proof}: 
Using Assumption \ref{assum_sojourntime} and the monotonicity property of function $f(x)=\frac{e^x-1}{x}$, one can establish that  $M+1$ is an upper pound of the PoA.
To show $M+1$ is also a lower bound, we give a sequence of instances  for which the ratio inside the supremum in \eqref{eq:PoA_definition} approaches $M+1$.
The complete proof can be found in Appendix~\ref{appendix:proof_prop:poa_m+1}.

Note that $\chi(\mathcal{G})\rightarrow1$ as $\beta\rightarrow0$,
which is consistent with PoA$=2$ in \cite{cd23} where agents maximize their long-run average reward. 
It is surprising that the PoA is independent of other parameters, especially the reward function. Instead, the sojourn time function plays an important role. 
\begin{corollary}\label{coro:infinite_poa}
For the game $\mathcal{G}$ with exponential discounting, the PoA is infinite. 
\end{corollary}

We can simply let $M\rightarrow\infty$ in Proposition \ref{prop:poa_m+1} to obtain  the corollary. A more detailed proof is given in Appendix \ref{appendix:proof_coro:infinite_poa}.
With exponential discounting the efficiency loss can be as large as possible. In fact, the worst case occurs with two actions in the constant action execution time model: the first action corresponds to the scenario when the agent receives a small amount of reward $r_1$ for a short duration of time $t_1$, while the second action corresponds to the scenario when the agent receives a large amount of reward $r_2$ for a long duration of time $t_2$. The PoA goes to infinity as both  $r_2$ and $t_2$ go to infinity. The myopic agents will clearly value action 1 more than action 2 due to exponential discounting. They ignore the action 2 which can give them an average reward which  we can choose to go to infinity. We will see that consistency of the time preferences is crucial in obtaining an infinite PoA in the following section.

Proposition \ref{prop:poa_m+1} has numerous implications. PoA is determined by the action with the largest $\chi_i(\mathcal{G})$ even if the $n$ parallel actions are independent of each other. It indicates  that  to increase system efficiency, we should improve the `worst action' with the largest  $\chi_i(\mathcal{G})$.   By doing so,  we can  at most reduce the price of anarchy to 2 since the minimal value of $\chi(\mathcal{G})$ is 1. Furthermore, to achieve the best possible price of anarchy, it may require reducing the sojourn time of the bottleneck action as much as possible,  taking into account that the discounting coefficient $\beta$ is usually fixed.   

\section{Power Law Discounting}\label{section:nonexponential_discounting}
In this section, we assume $u(r,t)=rt^{-\alpha}$ where $\alpha\ge 0$. We purposely allow $\alpha=0$ since it corresponds to the case of no discounting, i.e., when considering average reward. 
Since the Bellman's principle of optimality does not necessarily always hold, the best response of an agent may not be a stationary strategy. We will see that this is the case when $\alpha>1$. When $0\le\alpha\le1$, the best response of an agent can always be a stationary strategy and stationary equilibrium exists. Nevertheless,  we are able to establish PoA results whenever a stationary equilibrium exists. 


\subsection{Stationary vs Non-stationary Strategies}
 We wonder when an agent's best response against a stationary mass distribution can be a stationary strategy. This is necessary for the existence of a stationary equilibrium.
Since for every randomized stationary strategy there is a deterministic one which is at least as good (in terms of total utility), it suffices to only consider  deterministic stationary strategies. The total utility under a deterministic strategy which chooses action $i$ is
\[V(i,\boldsymbol{\tau})=\sum_{n=1}^\infty u(r_i,n\tau_i)
= r_i\tau_i^{-\alpha}\sum_{n=1}^\infty n^{-\alpha}\,,\]
which is finite if $ \alpha > 1$. If $0\le\alpha\le 1$, the series diverges and hence we compare the {\em total utility generated by time} $T$, given by the approximation 
\begin{equation}\label{eq:reward_VT}
 V_T(i, \boldsymbol{\tau}) =   \begin{cases}
r_i\tau_i^{-\alpha}\sum_{n=1}^{\left \lfloor \frac{T}{\tau_i} \right \rfloor}n^{-\alpha}\sim
\frac{r_i}{\tau_i}  T^{-\alpha+1},&\alpha<1\,,\\  
r_i\tau_i^{-1}\sum_{n=1}^{\left \lfloor \frac{T}{\tau_i} \right \rfloor}n^{-1}
\sim \frac{r_i}{\tau_i} \log{T},&\alpha=1\,,\\
\end{cases}
\end{equation}
as $T$ grows large. The condition $\alpha>1$ or $\alpha\le 1$ determines whether best response can always be a stationary strategy.
\begin{proposition}\label{prop:switchingornot}
    Consider game  $\mathcal{G}$ with power law discounting.
    \begin{enumerate}
        \item[1.] When $\alpha>1$, there exist a game instance of $\mathcal{G}$ and a stationary mass distribution $\boldsymbol{\mu}$ for which 
        the best response of an agent against $\boldsymbol{\mu}$ cannot be stationary.
        
        \item[2.] When $\alpha\leq 1$,  the best response of an agent against any stationary mass distribution can be
        a deterministic stationary strategy.
    \end{enumerate}
\end{proposition}

The proof can be found in Appendix \ref{appendix:proof_prop:switchingornot}.
When $\alpha>1$, the large discounting makes it more profitable for the agent to adopt a nonstationary strategy. 
To see this, with a slight loss of rigorousness, consider the constant action execution time model with two actions: the first one with a small reward and a  short sojourn time and the second one with a large reward and a long sojourn time.  The large discounting ratio makes it more profitable for the agents to choose action 1 in the beginning. As time goes on, their preference for immediate rewards diminishes, and it becomes more profitable for the agents to switch to action 2.

When $\alpha\le1$, the inconsistent time preferences do not have an important role; what is important is what happens in the distant future, due to the slow discounting. As a result, it is more profitable for the agent to adopt a stationary strategy. 
In the next section we show that a stationary equilibrium always exists if $\alpha \le 1$.

 \subsection{Existence of Stationary Equilibrium}

In the regime of low discounting, i.e., $\alpha\le1$, the best response against a mass distribution can always be  stationary. Using this fact, we are able to show that  there always exists a stationary equilibrium for the game $\mathcal{G}$ with $\alpha\le1$.

\begin{theorem}\label{thm:se_nonexpo}
There exists a stationary equilibrium for the game $\mathcal{G}$ with $\alpha\in[0,1]$.  $\boldsymbol{\mu}^\dagger$ is a stationary equilibrium if and only if 
\begin{equation}\label{eq:se_nonexpo_cond}
    \frac{r_i}{ \tau_i(\boldsymbol{\mu}^\dagger)}\ge \frac{r_j}{\tau_j(\boldsymbol{\mu}^\dagger)},\quad\forall\>i\in A^\dagger,\>j\in A,
\end{equation}
where $A^\dagger=\{i\in A\mid \mu^\dagger_i>0\}$.
\end{theorem}
\begin{corollary}\label{cor:se_nonexpo}
    The game $\mathcal{G}$ with $\alpha\in[0,1]$ is equivalent to a population game $\bar{\mathcal{F}}$, defined by the payoff vector,
\[\bar{F}(\boldsymbol{\mu})=\left(\frac{r_i}{ \tau_i(\boldsymbol{\mu})}\right)_{i\in A,}\]
i.e., the set of stationary equilibria of $\mathcal{G}$ coincides with the set of Nash equilibria of $\bar{\mathcal{F}}$. If Assumption \ref{assum_sojourntime} holds, $\bar{\mathcal{F}}$ is a stable game.
\end{corollary}

\emph{Sketch of Proof}: 
    By Proposition \ref{prop:switchingornot}, given any distribution $\bmu$,  there always exists a stationary strategy which is the best response agains $\bmu$. Similar to the proof of of Theorem \ref{thm:stationary_equilibrium}, one can represent mean field equilibria as fixed points of the best response correspondence and apply Kakutani's fixed point theorem to prove that stationary equilibrium exists.
   The proof of Corollary \ref{cor:se_nonexpo} is similar to  the proof of Proposition \ref{prop:stable_game}.

 In Appendix \ref{appendix:se_example_2} we utilize the characterization in Theorem \ref{thm:se_nonexpo} to
compute the equilibrium for the constant action execution time model. Note that $\alpha=0$ corresponds to the case of no discounting. Corollary \ref{cor:se_nonexpo} establishes an interesting result, that the regime of low discounting $(0<\alpha\le1)$ is equivalent to the  regime of no discounting if we are only concerned with stationary equilibria. When $\alpha>1$, it is not always true that a stationary equilibrium exists. 

\begin{proposition}\label{coro:se_nonexistence}
There does not always exist a stationary equilibrium for the game $\mathcal{G}$ with $\alpha>1$. More precisely, for any $\alpha>1$, there exists at least one instance of game  $\mathcal{G}$  for which there exists no stationary equilibrium, i.e., $\mathcal{S}(\mathcal{G})=\emptyset$.
\end{proposition}

The proof is given in Appendix \ref{appendix:proof_coro:se_nonexistence}.
The nonexistence of a stationary equilibrium poses challenges in analyzing system efficiency. Computing the social welfare, defined as the long-run average reward, becomes seemingly intractable. However, we can still compute the PoA in cases where a stationary equilibrium exists.

\subsection{Price of Anarchy of Stationary Equilibria}
In computing the PoA, we don't need to distinguish between the cases where $\alpha>1$ and $0\le \alpha\le1$ when stationary equilibria exist. 


\begin{proposition}\label{prop:poa=2}
Given Assumption \ref{assum_sojourntime} on the sojourn time vector $\boldsymbol{\tau}$,
 for all game instances $\mathcal{G}$ with  $\alpha\ge0$, 
 if a stationary equilibrium exists then the PoA  is 2, i.e.,
\[ \sup_{{\substack{\{(\mathcal{G}, \boldsymbol{\mu}^*, \boldsymbol{\mu}^\dagger) |\mathcal{S}(\mathcal{G})\not=\emptyset, 
    \boldsymbol{\mu}^* \in \mathcal{O}(\mathcal{G}),
\boldsymbol{\mu}^\dagger\in\mathcal{S}(\mathcal{G}) \}}}} 
    \frac{\text{SW}(\boldsymbol{\mu}^*)}{\text{SW}(\boldsymbol{\mu}^\dagger)}=2.\]
\end{proposition}

We have previously established in Corollary \ref{cor:se_nonexpo} that the low discounting regime ($0<\alpha\le1$) is equivalent to the regime of no discounting ($\alpha=0$). Hence, it is not surprising that $PoA=2$ for $\alpha\in(0,1]$, the same as the regime of no discounting. However,  the price of anarchy is still 2 in a high discounting regime ($\alpha>1$). Time-inconsistency plays a significant role in this. Myopic agents tend to prioritize immediate rewards. However, their inclination for immediate reward diminishes over time. Consequently, their actions and decision-making in the distant future may differ from their more short-sighted initial preferences. Power law discounting leads to the same behavior as no discounting in the distant future.
The mathematical explanation is that the power law decay is much slower than the exponential decay, which leads to an infinite PoA. 


\section{Learning Dynamics and Stability}\label{section:stability}
We investigate how discounting introduces new equilibria and affects their stability under certain learning dynamics in the context of a more general game model. 

We have seen in Proposition \ref{prop:stable_game} that under Assumption \ref{assum_sojourntime} the stateless mean field game $\mathcal{G}$ is a stable population game. As a result, all the mean field equilibria are globally asymptotically stable under many interesting dynamics, such as best response dynamics. This may not be the case when Assumption \ref{assum_sojourntime} fails. To explore this further, we consider a two-action and one-resource game model where both actions require access to a common resource, causing  externalities between the agents choosing different actions.
We will see that for this more general game model and no discounting, the mean field equilibrium is always locally asymptotically stable. But, under exponential discounting new equilibria appear due to myopic behavior, some of them not locally asymptotically stable.

\subsection{Two-Action and One-Resource Model}

\begin{figure*}
  \centering
  \subfigure[Dynamic System \eqref{dynamic:average} (Undiscounted)]{
    \label{figure:dynamic:stable} 
    \includegraphics[width=0.38\textwidth]{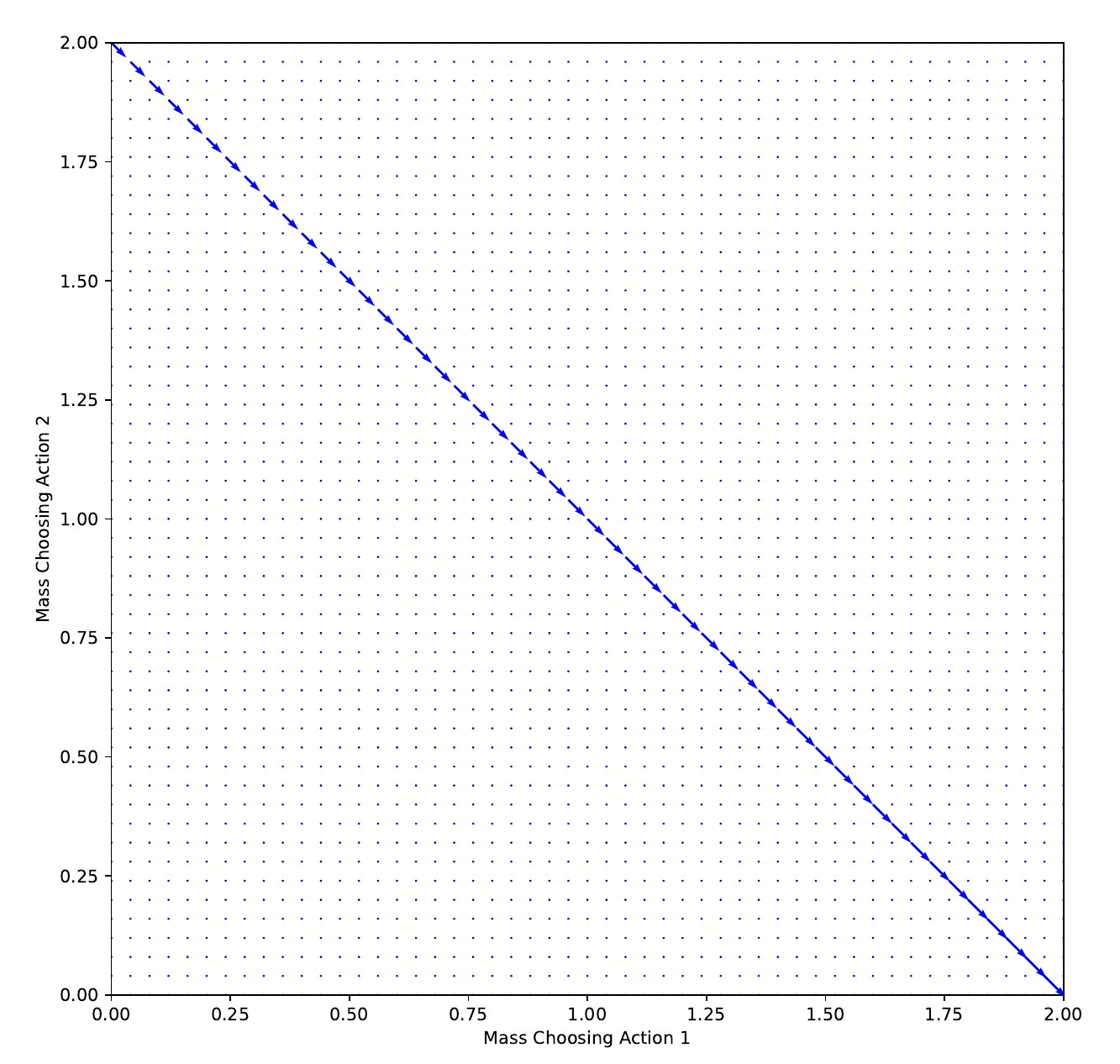}}
  \hspace{0.1\textwidth}
  \subfigure[Dynamic System \eqref{dynamic:discounted} (Discounted)]{
    \label{figure:dynamic:unstable} 
    \includegraphics[width=0.38\textwidth]{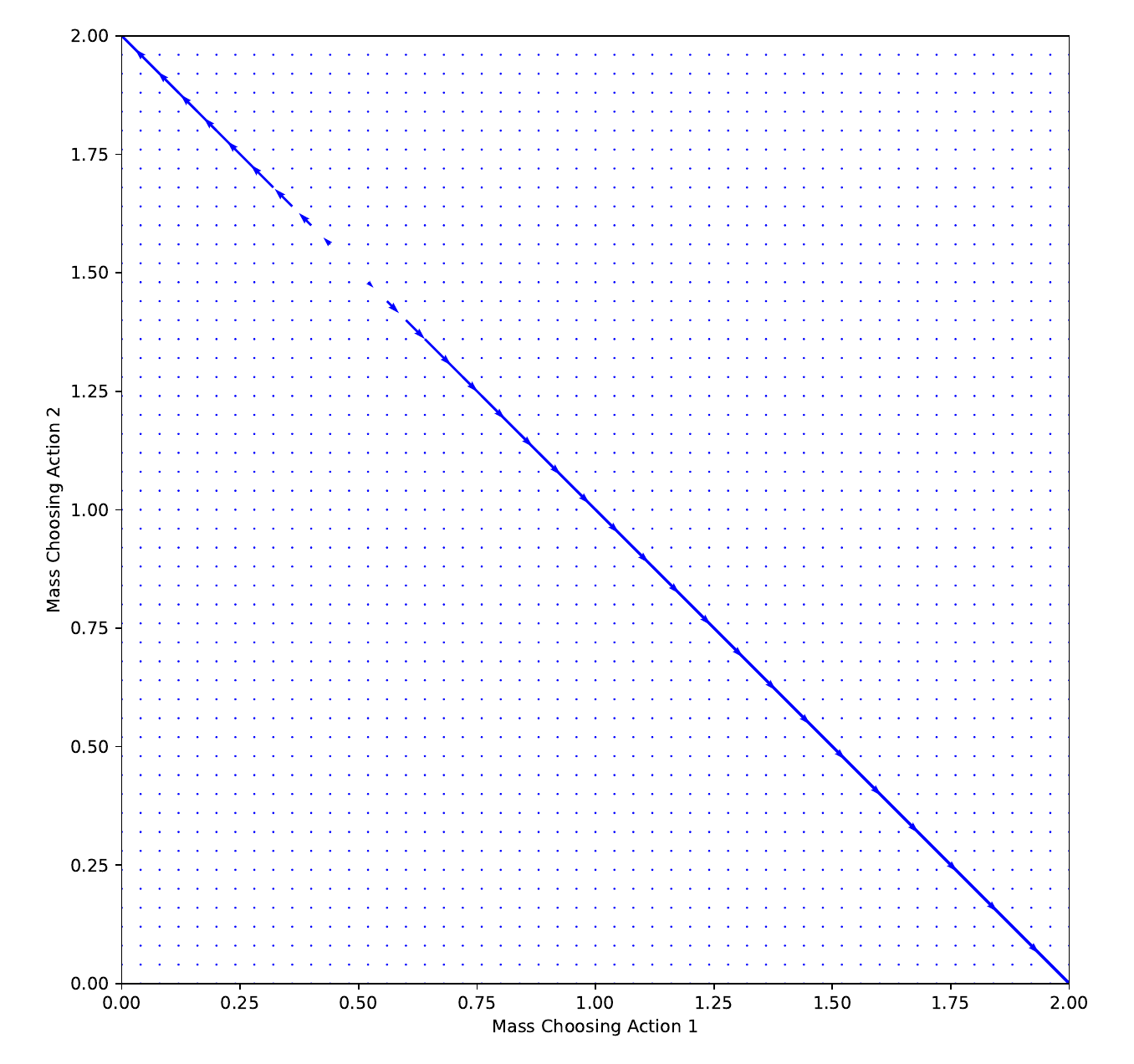}}
  \caption{Plots of the dynamics \eqref{dynamic:average} and \eqref{dynamic:discounted}.  The parameters are $m=2$,  $(t_1,t_2)=(3,0.5)$, $(r_1,r_2)=(e^5,e)$,  $(\gamma_1,\gamma_2)=(2,1)$, $b=1$  for both dynamics  \eqref{dynamic:average} and \eqref{dynamic:discounted}  and $\beta=1$  for \eqref{dynamic:discounted}.
  The direction of the arrow indicates the direction of mass flow: towards action 1 (lower left) or towards action 2 (upper right). 
  The length of the arrows represents the magnitude of the vector $(\dot\mu_1,\dot\mu_2)$ of the projection dynamics, and the rest points are the equilibrium points. 
  Under no discounting, there is a unique and stable equilibrium (the lower right corner) in the left figure. Under discounting, there are three equilibrium points in the right figure.
  We always keep the same stable equilibrium as in the undiscounted case, but introduce one more stable equilibrium (at the upper left corner) and an unstable equilibrium (in between the two).
  }
  \label{figure:dynamic} 
\end{figure*}

We consider a stateless mean field game with two actions, i.e., $A=\{1,2\}$. The respective execution times  are  $t_1,t_2>0$ (constant), and the immediate rewards are $r_1,r_2$. 
We assume the two actions require $\gamma_1$ and $\gamma_2$ units of resources per unit time from a common resource supplied with rate $b>0$.
Similar to the convex optimization problem \eqref{convex_optimization}  , we can compute the sojourn times $(\tau_1,\tau_2)$ by first solving the following optimization problem
\begin{align}
\max &\quad \mu_{1} \log x_{1}  + \mu_{2} \log x_{2} - t_1x_1-t_2x_2  \label{2a1r:convex_optimization} \\
\text{s.t. } &\quad   \gamma_1x_1+\gamma_2x_2\le b,\label{2a1r:resource_constraint}\\    
\text{over }&\quad (x_1,x_2)\in\mathbb{R}^{2}_+,\nonumber
\end{align}
where the waiting time $w$, which is the same for both actions, is the optimal Lagrange multiplier for the resource constraint \eqref{2a1r:resource_constraint}. Then, the sojourn time $\tau_i$ is explicitly given by $\tau_i=t_i+\gamma_iw$ for $i\in A$. Since $\mu_2=m-\mu_1$, $\tau_i$ is a function of $\mu_1$ for $i=1,2$.

By Proposition \ref{prop:stable_game}, this game is equivalent to a population game defined by the payoff vector
\[\bold{F}^\beta(\mu_1)=(F_1^\beta(\mu_1),F_2^\beta(\mu_1))=\left(\frac{r_1}{e^{\beta\tau_1(\mu_1)}-1},\frac{r_2}{e^{\beta\tau_2(\mu_1)}-1}\right)\,.\]
We denote this two-action one-resource game with exponential discounting by $\cFbeta$. $\mu^\dagger_1$ is an equilibrium of $\cFbeta$ if $F_1^\beta(\mu^\dagger_1)=F_2^\beta(\mu^\dagger_1)$.

To compare, we  consider the  two-action one-resource game with average reward criterion $\barcF$, defined by the payoff vector
\[\bar{\bold{F}}(\mu_1)=(\bar{F}_1(\mu_1),\bar{F}_2(\mu_1))=\left(\frac{r_1}{\tau_1(\mu_1)},\frac{r_2}{\tau_2(\mu_1)}\right).\]

We will see that any equilibrium of $\barcF$ (undiscounted) is locally asymptotically stable, while ${\cFbeta}$ (discounted) can have additional equilibria, some of them unstable.

\subsection{Projection Dynamics}

For the game $\cFbeta$ (with exponential discount), we consider the projection dynamics \cite{Zhang1996} defined by
\begin{equation}\label{dynamic:discounted}
    \dot \mu_1 =
    \begin{cases}
        \frac{r_1}{e^{\beta\tau_1(\mu_1)}-1}-\frac{r_2}{e^{\beta\tau_2(\mu_1)}-1}&\text{ if }\mu_1\in(0,1),\\
        \max\{\frac{r_1}{e^{\beta\tau_1(\mu_1)}-1}-\frac{r_2}{e^{\beta\tau_2(\mu_1)}-1},0\}&\text{ if }\mu_1=0,\\
        \min\{\frac{r_1}{e^{\beta\tau_1(\mu_1)}-1}-\frac{r_2}{e^{\beta\tau_2(\mu_1)}-1},0\}&\text{ if }\mu_1=1\,,\\
    \end{cases}
\end{equation}
 and for the game $\barcF$ without discounting
\begin{equation}
    \dot \mu_1 =\begin{cases}
        \frac{r_1}{\tau_1(\mu_1)}-\frac{r_2}{\tau_2(\mu_1)}&\text{ if }\mu_1\in(0,1),\\
        \max\{ \frac{r_1}{\tau_1(\mu_1)}-\frac{r_2}{\tau_2(\mu_1)},0\}&\text{ if }\mu_1=0,\\
        \min\{ \frac{r_1}{\tau_1(\mu_1)}-\frac{r_2}{\tau_2(\mu_1)},0\}&\text{ if }\mu_1=1\,.\\
    \end{cases}\label{dynamic:average}
\end{equation}

The dynamics for $\mu_2$ are omitted since $\mu_1+\mu_2=m$. It is clear that under these dynamics, the agents' collective movement transitions from the action with lower reward to the one with higher reward. Furthermore, they require that the absolute growth rate of the mass choosing one action to be the difference between its payoff and the payoff of the other action. We are mainly concerned with the stability of the Nash equilibria with or without discounting. We first prove that the Nash equilibrium of $\barcF$ is always stable. 

\begin{proposition}\label{prop:stabiltity}
    Every equilibrium of the game $\barcF$ with average reward criterion is locally asymptotically stable under the  dynamics in \eqref{dynamic:average}. 
\end{proposition}

 The proof is given in Appendix \ref{appendix:proof_prop:stability}.  Proposition \ref{prop:stabiltity} does not hold for the game  $\cFbeta$. The local stability of  $\cFbeta$ depends on the parameters $m$, $(t_1, t_2)$, $(r_1,r_2)$, $(\gamma_1,\gamma_2)$, $b$, and $\beta$. More specifically, we have the following results regarding when $\cFbeta$ can have unstable equilibrium. 
 \begin{proposition}\label{prop:unstable}
     Consider any $m>0$, $\mu^\dagger_1\in(0,m)$ and $(t_1, t_2, r_1,$
     $r_2, \gamma_1,\gamma_2, b, \beta)\in\mathbb{R}_+^8$  satisfying the following conditions
      \begin{equation}
      \frac{\gamma_1}{\gamma_2}>\frac{t_1}{t_2},\ \frac{\gamma_1}{\gamma_2}<\frac{e^{\beta t_1}-1}{e^{\beta t_1}}\quad\text{or}\quad\frac{\gamma_1}{\gamma_2}<\frac{t_1}{t_2},\ \frac{\gamma_1}{\gamma_2}>\frac{e^{\beta t_2}}{e^{\beta t_2}-1},\label{condition:unstable_1}
      \end{equation}
      and
        \begin{equation}
    \frac{r_1}{e^{\beta \tau_1(\mu_1^\dagger)}-1}=\frac{r_2}{e^{\beta \tau_2(\mu_1^\dagger)}-1},\quad b<\gamma_1\frac{
\mu_1^\dagger}{t_1}+\gamma_2\frac
{m-\mu_1^\dagger}{t_2}\,,\label{condition:unstable_2}
      \end{equation}
      where $\tau_1$, $\tau_2$ are functions of $\mu_1$ implicitly defined by \eqref{2a1r:convex_optimization} with parameters $t_1$, $t_2$, $\gamma_1$, $\gamma_2$, and $b$. Then,
 $\mu^\dagger_1$ is an unstable equilibrium of the game $\cFbeta$ under the dynamics \eqref{dynamic:discounted}.
 \end{proposition}

 The proof is given in Appendix \ref{appendix:proof_prop:unstable}. 
 Proposition \ref{prop:unstable} identifies a nonempty set of parameters for which the game $\cFbeta$ has an unstable equilibrium  under the projection dynamic. In fact, the game $\cFbeta$ may have multiple equilibria,  with some being stable and others unstable. This will be shown through a specific example in the next subsection.

\subsection{Example of Unstable Equilibrium}

In this section, we provide an example of two-action one-resource game for which there is a unique stable game when considering average reward (game $\barcF$), while there are three equilibria with two being stable and one being unstable when considering discounted reward (game $\cFbeta$). The parameter are given as follows, $m=2$, , $(t_1,t_2)=(3,0.5)$, $(r_1,r_2)=(e^5,e)$,  $(\gamma_1,\gamma_2)=(2,1)$, $b=1$  for both games $\barcF$ and $\cFbeta$,
and $\beta=1$  for game $\cFbeta$.

Fig. \ref{figure:dynamic} plot the dynamics \eqref{dynamic:average} and \eqref{dynamic:discounted}, including the dynamics of $\mu_2$, see the explanations in the caption of the figure. 
There is a unique equilibrium $\mu_1^\dagger=m=2$ for the game $\barcF$ as in Fig. \ref{figure:dynamic:stable}. It is stable, since a perturbation to the equilibrium will decay to zero according to the direction of the arrows near the equilibrium point.

We observe in Fig. \ref{figure:dynamic:unstable} that there are three equilibria for game $\cFbeta$. The two extreme equilibria, $\mu_1^\dag=0$ and $\mu^\dag_1=m=2$, are stable, and the interior equilibrium   $\mu^\dagger_1\approx0.5$ is unstable.  $\mu^\dagger_1 = 0$ is a stable equilibrium because \emph{the sojourn time is the primary factor} in determining the relative payoff between the two actions. Agents tend to choose the action with the smaller action time (Action 2), and  small perturbations do not change the ordering of these times. This equilibrium exists \emph{due to discounting}, since the reward is collected after the end of the action.

Similarly,  $\mu^\dagger_1 = m=2$ is a stable equilibrium because the \emph{immediate reward is the primary factor}  in determining the relative payoff between the two actions. Agents tend to choose the action with higher immediate reward (Action 1), and slight perturbations do not change this property.

The interior equilibrium at   $\mu^\dagger_1 \approx 0.5$ is where the effects of time and immediate reward balance out. However, this balance is unstable, indicating that small perturbations can make one of these effects dominant, leading to the corresponding stable equilibrium.

\section{Conclusions}

Overall, our research highlights  for the first time the intricate relationship between time preference, existence of equilibrium, stability of equilibria, and the price of anarchy in MFGs  where action externalities are due to queuing for shared resources, and agents are myopic.

We have focused on a specific class of MFGs where nonatomic  myopic agents compete for congestible resources, and interactions among these agents are continuous and temporal. Our study examines the impact on the stationary equilibria of myopic decision-making with different time preference, i.e., exponential vs power law discounting. We found that the existence, efficiency and stability of equilibria highly depend on the types of discounting  under consideration.
Under normal exponential discounting, a stationary equilibrium always exists and PoA$=\infty$, compared to PoA$=2$ under no discounting, indicating significant inefficiencies resulting from myopic decision making. Under power law discounting, the existence of a stationary equilibrium is not guaranteed, but, if such an equilibrium does exist, PoA$=2$,  similar to the undiscounted case. This result suggests that time inconsistency in discounting may not always be detrimental, as in the case of exponential discounting.
Finally, we also found that  discounting can introduce unstable equilibrium in the case of agents choosing among different tasks whose execution times are correlated. 

There are many  immediate extensions to this research. We propose exploring the PoA results for a general Markov decision process model that includes state transitions with the available actions depending on the state of the system, as our current model does not allow it. 
This extension could enhance the applicability of our PoA analysis. A remaining question on the theory side is to find a constructive proof for the PoA with exponential discounting for a system with multiple states (i.e., a general finite state MDP) as in \cite{cd23}.  We also recommend examining a more general sojourn time model where the assumption of monotonicity, as stated in Assumption \ref{assum_sojourntime}, may not be relevant. 
Furthermore, we would like to analyze the stability properties and the performance of different learning dynamics in the previously mentioned, more general models.

 \newpage

\bibliographystyle{splncs04}
\bibliography{biblio-1}

\appendix

\section{Motivating Example}\label{Motivating Example}

We consider a system where a common pool of agents choose between two independent actions  when they are free as depicted in Fig.~\ref{fig:twoaction}. We refer to this model as \emph{`two-action model'} in what follows.
 In the context of food delivery, this corresponds to a network with two hotspots, where couriers decide between two shopping mall areas. Once a courier becomes idle, they must decide which hotspot to relocate to in order to serve new orders.

There are two loops in Fig.~\ref{fig:twoaction} and each corresponds to a single action $i\in\{1,2\}$.  We first consider a simpler system consisting of only one loop (say the upper loop with action $i=1$) to define the queueing model when there is a single action 
related to a single resource that gets consumed by the action. 
In this case, the time that it takes for an agent to execute the action, i.e., the sojourn time of the action $\tau_i$, is the sum of the waiting time in the resource queue $w_i$ plus the time to traverse the delay element that correspond to the 
intrinsic time of the action $t_i$.
Little's Law $\mu_i=x_i\tau_i$ relates the actual rate $x_i$ of agents performing the action with the number of agents $\mu_i$ that choose the action and the total time $\tau_i$ in the system. 
In equilibrium, the rate $x_i$ can not exceed the supply rate $b_i$ of the resource, i.e., $x_i\leq b_i$.

Let's understand how $\tau_i,x_i$ change as $\mu_i$ increases for any $i\in\{1,2\}$. In the beginning, for small $\mu_i$, there is no queue, the sojourn time equals the delay element $t_i$ and $x_i$ satisfies $\mu_i=x_i t_i$ with $x_i<b_i$. Observe that 
$x_i$ increases linearly with $\mu_i$ until it reaches the supply rate $b_i$.
When $\mu_i\geq b_it_i$, agents queue: a mass $b_it_i$ of them is in the delay element and $\mu_i-b_it_i$ in the queue before entering the delay element. Now the sojourn time satisfies the equation $\mu_i=b_i \tau_i$ where $\tau_i=w_i+t_i$, where $w_i$ is the delay to access the resource. Observe that $\tau_i$ is increasing in $\mu_i$ and agents obtain a reward $r_i$ after exiting the delay element, i.e., after spending a total of $\tau_i$ in the system. 


\begin{figure*}[t]
    \centering
  \includegraphics[width=4in]{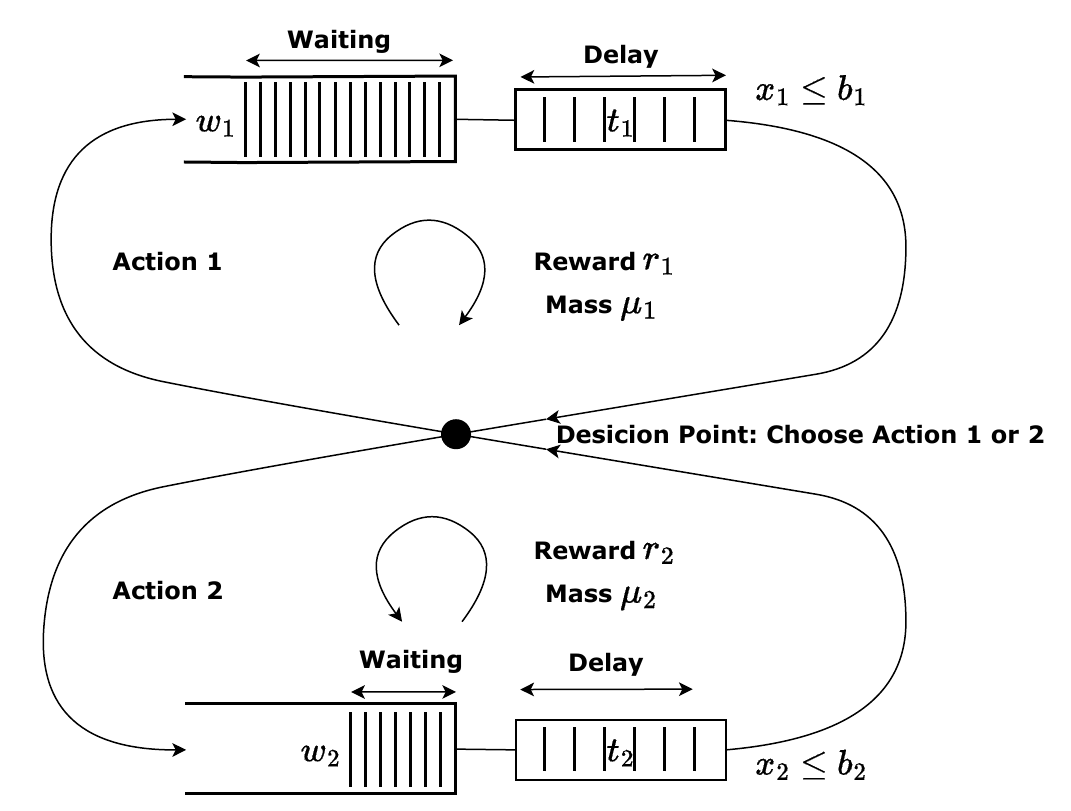}
    \caption{A system where agents choose between two independent actions when they are not busy. 
    For each action $i\in\{1,2\}$, the time that it takes for an agent to execute the action, i.e., the sojourn time of the action $\tau_i$, is the sum of the waiting time in the resource queue $w_i$ plus the time to traverse the delay element that correspond to the 
intrinsic time of the action $t_i$. $\mu_i$  is the number of agents performing action $i$, and it satisfies Little's Law $\mu_i=x_i\tau_i$ where $x_i$ is the rate of agents performing action $i$.  In equilibrium, the rate $x_i$ can not exceed the supply rate $b_i$ of the resource, i.e., $x_i\leq b_i$.  Under no discounting, agents should pursue the action with the largest average reward per unit time $\frac{r_i}{\tau_i}$ unless both actions generate the same average reward.}
    \label{fig:twoaction}
\end{figure*}

 Now we consider the two loops in Fig.~\ref{fig:twoaction}  as a unified system. Assume a total mass $m$ of agents choose between  the two actions and interact continuously. 
A non-myopic agent will choose the action $i\in\{1,2\}$ with the highest reward rate per unit of time, $\frac{r_i}{t_i+w_i}$, where $r_i$ is the reward for performing action $i$,  $t_i$ is the corresponding sojourn time, and $w_i$ is the waiting time in the resource queue, which depends on the mass of agents $\mu_i$ that choose action $i$. The choices made by the agents determine the mass of agents in the queues and in the delay elements depicted in Fig.~\ref{fig:twoaction}. This establishes the dynamics as follows: given any distribution $\mu_1,\mu_2$ of agents ($\mu_1+\mu_2=m$), agents will dynamically determine the action  offering the highest reward per unit time and may adjust their strategies, leading to a shift in the mass distribution of agents. This adjustment will influence the action sojourn times, subsequently affecting the associated rewards and will persist until an equilibrium is attained, where both actions yield identical average rewards.
Under a similar model that allows for more general resource interactions, \cite{cd23} shows the existence of an equilibrium between decisions and waiting times, with a PoA of 2, \emph{when there is no time discounting for the agents}.

Suppose now the agents are myopic, using exponential discounting with factor $e^{-\beta t}$, where $\beta>0$. We would like to answer the following questions.
\begin{itemize}
    \item In view of the dynamic nature of the game, does there exist a stationary equilibrium where drivers use 
 time independent policies?
 \item  How much efficiency is lost due to decentralized myopic behavior, from the perspective of a social planner who values long-term average rewards?
\end{itemize}
 For this two-action example, an agent will choose the action $i$ yielding the highest total discounted reward, $\frac{r_i}{e^{\beta(t_i+w_i)}-1}$. As previously discussed, the challenging aspect of the problem lies in the fact  that agents' choices are highly sensitive to waiting delay, and waiting delay is determined by the aggregated effect of the agents' choices .
Using a fixed-point argument, in Theorem~\ref{thm:stationary_equilibrium} we show that an stationary equilibrium always exists. It can be shown that the PoA is $\infty$ (Corollary~\ref{coro:infinite_poa}). The most inefficient system is actually an extreme case of the simple two-action model discussed here. More specifically, in Figure ~\ref{fig:twoaction}, assume that action 1 offers a small reward after a short sojourn time, while action 2 provides a larger reward but after a longer sojourn time. Due to their myopic nature, agents tend to prioritize action 1 over action 2.
The exponential discounting ensures that time preference remains consistent over time. However, this can lead to an extreme scenario where action 2 becomes infinitely more profitable than action 1 in terms of long-run average reward.

We also consider the case of the `power law' discounting factor $t^{-\alpha}$, where $\alpha > 0$. The existence of stationary equilibrium under exponential discounting relies on the fact that the time preference of the agents are consistent, i.e., they follow a consistent pattern of discounting value over time, which makes it possible for a stationary policy to be optimal. Under power law discounting, inconsistency of time preference may cause a problem. We  would like to answer the following questions.
\begin{itemize}
    \item  Does stationary equilibrium still exist when agents are inconsistent in terms of the ways they discount the reward over time?
 \item  Does the price of anarchy remain infinite, or is it possible that the inconsistent time preferences of agents could enhance efficiency?
\end{itemize}

Even for this simple example, it is challenging to find optimal policies for myopic agents under power law discounting as it is not necessarily time independent. However, we still managed to prove a surprising result that if $\alpha \leq 1$, the equilibrium is the same as in the nonmyopic case, giving a PoA equal to 2 (Theorem~\ref{thm:se_nonexpo}). The case $\alpha > 1$ is more involved because an equilibrium may fail to exist. (That is, it does not {exist} within the class of stationary policies. {In certain situations, a nonstationary policy is better than any stationary one}). Nevertheless it is shown that whenever an equilibrium exists, it is equal to 2.
Thus, whenever equilibria exist in the power law case, the PoA is 2, regardless of the value of $\alpha$ (Proposition~\ref{prop:poa=2}). As we mentioned in the introduction, under power law discounting, the preference for immediate reward is not consistent over time. When evaluating the long-term average reward, it's notable that the efficiency loss is bounded, and even more surprisingly,  the same as the no discounting case.  

Fig.~\ref{fig:twoaction} presents a model where each action requires a single independent resource and hence a single indepenedent queue is formed for each action. This may not be the case in practice. In reality, a single action may require multiple resources and multiple actions may share overlapping resources. 
The underlying queueing network can be quite complex. Nevertheless, we successfully established 
all the theoretical results outlined above  under fairly general assumptions (Assumption \ref{assum_sojourntime}) on the underlying queueing system, or equivalently, the dependence of sojourn time on the agent distribution. However, the simple two-action model presented here remains interesting 
as we will see that we obtain the most inefficient system by considering extreme cases of the two-action model for both exponential discounting and power law discounting.

\section{ Fluid Sojourn Time Model}\label{appendix:fluid_model}
We consider a resource model with independent actions and parallel resources, where  each action $i$ consumes one unit of resource $l_i$ per unit time. The replenishment rate of resource $l_i$ is $b_i$. The resources $l_i$'s,  $i\in A$, are `in parallel', i.e., independent of each other. 
Assume for each action $i$, there is an action execution time for this action to be completed,  which is an increasing function of the rate $x_i$, denoted by $t_i(x_i)$.
Given any mass distribution $\bmu$, because of the resource constraint,  we have
\begin{equation}\label{ineq:upper_bound}
    x_i\le b_i,\quad\forall\>i\in A\,.
\end{equation}
 as  now $x_i(\bmu)$ is  also the consumption rate of resource $l_i$.

The sojourn time $\tau_i$ is comprised by the waiting time to acquire the resource and the 
actual execution time:
\begin{equation}\label{eq:decomposition}
     \tau_i=w_i+t_i(x_i),\quad\forall\>i\in A.
\end{equation}

Now, in a fluid queueing model no queues are formed if the resource constraint is not active, i.e.,
\begin{equation}\label{ineq:strict}
    x_i< b_i\implies w_i=0,\quad\forall\>i\in A.
\end{equation}

It can be shown that for any $\boldsymbol{\mu}\in \mathcal{U}_m(A)$, the conditions \eqref{eq:little_law}, \eqref{eq:decomposition}, \eqref{ineq:upper_bound}, \eqref{ineq:strict} uniquely determine $x_i$ and $w_i$ for any $i\in A$. 

\begin{lemma}\label{lemma:sojourn_time}
Conditions  \eqref{eq:little_law}, \eqref{eq:decomposition}, \eqref{ineq:upper_bound}, \eqref{ineq:strict} define a unique continuous mapping from $\mathcal{U}_m(A)$ to $\mathbb{R}^n\times\mathbb{R}^n\times\mathbb{R}^n$ by
\[\boldsymbol{\mu}\mapsto (\bx(\boldsymbol{\mu}),\bw(\boldsymbol{\mu}),\boldsymbol{\tau}(\boldsymbol{\mu})),\]
where $\bx(\boldsymbol{\mu})=(x_i(\boldsymbol{\mu}))_{i\in A}$, $\bw(\boldsymbol{\mu})=(w_i(\boldsymbol{\mu}))_{i\in A}$, $\boldsymbol{\tau}(\boldsymbol{\mu})=(\tau_i(\boldsymbol{\mu}))_{i\in A}$. Moreover, $\bx(\boldsymbol{\mu})$ is the unique solution to the following convex optimization problem
\begin{align}
\max\quad & \sum_{i\in A} {\mu}_{i} \log(x_{i})   - \sum_{i\in A}\int_0^{x_i}t_i(y)\dif y \label{convex_optimization_A} \\
\text{s.t. } &   x_i\le b_i,\quad\forall\>i\in A,\label{constraint:resource_i_A}\\    
\text{over }& \bx = (x_{i})_{i\in A}\in\mathbb{R}^{n}_+,\nonumber
\end{align}
and $w_i(\boldsymbol{\mu})$'s are the optimal Lagrange multipliers for the resource constraints  \eqref{constraint:resource_i_A}.
\end{lemma}

\begin{proof}

First we show that the mapping is well defined.
Given any $\boldsymbol{\mu}\in\mathcal{U}_m(A)$, the optimization problem
\begin{align}
\max\quad & \sum_{i\in A} \mu_{i} \log(x_{i})   - \sum_{i\in A}\int_0^{x_i}t_i(y)\dif y  \nonumber\\ 
\text{s.t. } &   x_i\le b_i,\quad\forall\>i\in A,\nonumber\\    
\text{over }& \bx = (x_{i})_{i\in A}\in\mathbb{R}^{n}_+.\nonumber
\end{align}
has a unique solution as the objective is a strictly concave function maximized over a set 
with compact closure and $x_{i}>0$ unless $\mu_{i} = 0$.  Strong duality holds and the KKT conditions characterizing the optimal solution, with $w_i, i = 1,\ldots,n$ being the optimal Lagrange multipliers, are equivalent to \eqref{eq:little_law}, \eqref{eq:decomposition},  \eqref{ineq:upper_bound}, \eqref{ineq:strict}.

The mapping $\boldsymbol{\mu}\mapsto \bx$ is continuous because the objective is continuous in $\boldsymbol{\mu}$ and $\bx$ is unique. The
continuity of $w$ follows from~\eqref{eq:little_law}, \eqref{eq:decomposition}. 

\end{proof}

\section{Proof of Theorem \ref{thm:stationary_equilibrium}}\label{appendix:proof_thm:stationary_equilibrium}
 For any feasible mass distribution $\boldsymbol{\mu}$, define
    \begin{align}
        v(\boldsymbol{\mu})&=\max_{i\in A}V(i,\boldsymbol{\mu})=\max_{i\in A}\frac{r_i}{e^{\beta \tau_i(\boldsymbol{\mu})}-1},\label{eq:value_function_exponential}
    \end{align}
    Note that $v(\boldsymbol{\mu})$ satisfies the following Bellman equation
    \begin{equation}\label{eq:Bellman}
        v(\boldsymbol{\mu})=\max_{i\in A}e^{-\beta \tau_i(\boldsymbol{\mu})}r_i+e^{-\beta \tau_i(\boldsymbol{\mu})}v(\boldsymbol{\mu}).
    \end{equation}

    We first show that given any agent mass distribution $\boldsymbol{\mu}$, the supremum in \eqref{condition_optimal_strategy} is achieved by a  deterministic stationary strategy.
    Let
    \[V^*(\boldsymbol{\mu})=\sup\limits_{\boldsymbol{\sigma}\in \mathcal{P}(A)^{\mathbb{N}}}\mathbb{E}_{\boldsymbol{\sigma}}\left[\sum_{k=1}^\infty u\bigg(r_{a_k},\sum_{j=1}^k\tau_{a_j(\boldsymbol{\mu})}\bigg)\right].\]
    For any strategy $\boldsymbol{\sigma}\in\mathcal{P}(A)^{\mathbb{N}}$ we have
    \begin{align*}    &V(\boldsymbol{\sigma},\boldsymbol{\mu})\\
&\quad=\mathbb{E}_{\boldsymbol{\sigma}}\left[\sum_{k=1}^\infty r_{a_k}e^{-\beta \sum_{j=1}^k\tau_{a_j}(\boldsymbol{\mu})}\right]\\
    &\quad=\sum_{i=1}^n\sigma_1(i)\Bigg(r_ie^{-\beta\tau_i(\boldsymbol{\mu})}+\mathbb{E}_{\boldsymbol{\sigma}}\left[\sum_{k=2}^\infty r_{a_k}e^{-\beta \sum_{j=1}^k\tau_{a_j}(\boldsymbol{\mu})}\right]\Bigg)\\
    &\quad=\sum_{i=1}^n\sigma_1(i)\Bigg(r_ie^{-\beta\tau_i(\boldsymbol{\mu})}+e^{-\beta\tau_i(\boldsymbol{\mu})}\mathbb{E}_{\boldsymbol{\sigma}}\left[\sum_{k=2}^\infty r_{a_k}e^{-\beta \sum_{j=2}^k\tau_{a_j}(\boldsymbol{\mu})}\right]\Bigg)\\
    &\quad\le \sum_{i=1}^n\sigma_1(i)\bigg(r_ie^{-\beta\tau_i(\boldsymbol{\mu})}+e^{-\beta\tau_i(\boldsymbol{\mu})}V^*(\boldsymbol{\mu})\bigg)\\
    &\quad\le \max_{i\in A}\{r_ie^{-\beta\tau_i(\boldsymbol{\mu})}+e^{-\beta\tau_i(\boldsymbol{\mu})}V^*(\boldsymbol{\mu})\}.
    \end{align*}
    Since $\boldsymbol{\sigma}$ is an arbitrary strategy, it follows that
    \[V^*(\boldsymbol{\mu})\le \max_{i\in A}\{r_ie^{-\beta\tau_i(\boldsymbol{\mu})}+e^{-\beta\tau_i(\boldsymbol{\mu})}V^*(\boldsymbol{\mu})\}.\]
    Thus, for some $j\in A$ we have
    \[V^*(\boldsymbol{\mu})\le\frac{r_j}{e^{\beta\tau_j(\boldsymbol{\mu})}-1}\le v(\boldsymbol{\mu})=\max_{i\in A}V(i,\boldsymbol{\mu})\le V^*(\boldsymbol{\mu}).\]
    Therefore,
    \[V^*(\boldsymbol{\mu})=v(\boldsymbol{\mu})=\max_iV(i,\boldsymbol{\mu}).\]
    In other words, the optimal expected total discounted reward can be obtained by a deterministic stationary  strategy.

Now, define
        \begin{align}
          A^\ast(\boldsymbol{\mu})&=\argmax_{i\in A}\frac{r_i}{e^{\beta \tau_i(\boldsymbol{\mu})}-1},\label{eq:optimal_action_set}\\
         D(\boldsymbol{\mu})&=\{
         e_i\mid i\in A^\ast(\boldsymbol{\mu})\},\\
          \text{conv}(D(\boldsymbol{\mu}))&=\{\sigma\in\mathcal{P}(A)\mid \sigma_i=0,\>\forall\>i\notin 
         A^\ast(\boldsymbol{\mu})\},\\
          \text{mconv}(D(\boldsymbol{\mu}))&=\{\boldsymbol{\mu}^\prime\in\mathcal{U}_m(A)\mid \boldsymbol{\mu}^\prime_i=0,\>\forall\>i\notin 
         A^\ast(\boldsymbol{\mu})\}.
    \end{align}
$D(\boldsymbol{\mu})$ is the set of optimal deterministic stationary  strategies. $\text{conv}($
$D(\boldsymbol{\mu}))$ is the set of optimal stationary  strategies. $ \text{mconv}(D(\boldsymbol{\mu}))$ is the set of stationary mass distributions associated with the set of optimal stationary strategies. 

Let $\phi$ be a set-valued map defined as  follows
\begin{equation}
 \phi(\boldsymbol{\mu})=\text{mconv}(D(\boldsymbol{\mu})).\label{bestresponse}   
\end{equation}
The existence of stationary equilibrium is equivalent to the existence of a fixed point of the map $\phi$. We show the latter by Kakutani's fixed point theorem.

Note that $\phi(\boldsymbol{\mu})$ is obviously nonempty and convex for any feasible mass distribution $\boldsymbol{\mu}\in \mathcal{U}_m(A)$. Then, we will show that the map $\phi$ is closed. We only need to show for any sequence $(\bz^k)_{k\in\mathbb{N}}\in \mathcal{U}_m(A)^\mathbb{N}$ and $(\by^k)_{k\in\mathbb{N}}\in \phi(\bz^k)$ with $\lim_{n\rightarrow\infty}\bz^k=\bz$ and  $\lim_{k\rightarrow\infty}\by^k=\by$, we have $\by\in\phi(\bz)$.  For any $k\in\mathbb{N}$, let $\sigma^k$ be the strategy compatible with the distribution $\by^k$ as in Definition \ref{def:stationary_ne}, that is,
\[\sigma^k(i)=\frac{\frac{y^k_i}{\tau_i(\by^k)}}{\sum_{j=1}^{n}\frac{ y^k_j}{\tau_j(\by^k)}},\quad\forall\>i\in A.\]
Since $\sigma^k(i)>0$ if and only if $y^k_i>0$ for any $i\in A$, $\sigma^k$ is an optimal strategy for the agent given distribution $\bz^k$, or equivalently, $\sigma^k\in\text{conv}(D(\bz^k))$ for any $k\in\mathbb{N}$.
Now define $\sigma^\by$ as 
\[\sigma^\by(i)=\frac{\frac{y_i}{\tau_i(\by)}}{\sum_{j=1}^{n}\frac{ y_j}{\tau_j(\by)}},\quad\forall\>i\in A.\]
Similarly, the distribution $\by$ is compatible with the strategy $\sigma^\by$ and $\sigma^\by(i)>0$ if and only if $y_i>0$ for any $i\in A$. Since the function $\tau_j$ is assumed to be continuous,  $\lim_{k\rightarrow\infty}\sigma^k=\sigma^\by$. Note that $V(\sigma,\boldsymbol{\mu})$ is continuous in both $\sigma$ and $\boldsymbol{\mu}$. Additionally, $V^*(\boldsymbol{\mu})$ is continuous in $\boldsymbol{\mu}$.  Then,
\[V(\sigma^\by,\bz)=\lim_{k\rightarrow\infty}V(\sigma^k,\bz^k)=\lim_{k\rightarrow\infty}V^*(\bz^k)=V^*(\bz).\]
Note that for any $i\in A$, $\sigma^\by(i)>0$ only if $V(i,\bz)=V^*(\bz)$ since
\[V(\sigma^\by,\bz)=\sum_{i=1}^n\sigma^\by(i)V(i,\bz).\]
It follows that $\sigma^\by\in\text{conv}(D(\bz))$ and hence $y\in\phi(\bz)$.

We have shown that $\phi$ is nonempty, convex and closed. It follows from Kakutani's fixed point theorem that $\phi$ has at least one fixed point and the first part of the theorem is proven. 

To prove the second part of the theorem, we only need to note that  if 
 $\bmu^\dagger$ is an equilibrium, $i\in A^\dagger$ is the optimal deterministic strategy against $\bmu^\dagger$. Then, for any $i\in A^\dagger$,
 $$V(i,\bmu)=\frac{r_i}{e^{\beta \tau_i(\boldsymbol{\mu})}-1}$$ 
 is the optimal value, i.e., $V(i,\bmu)\ge V(j,\bmu)$ for any $i\in A^\dagger$ and $j\in A$.

\section{ The Constant Action Execution Time Model under Exponential Discounting}\label{appendix:se_example_1}

We consider the constant action execution time model.
Specifically, we consider game instances $\mathcal{G}$ with $\btau$ given by the fluid queueing model introduced in Appendix \ref{appendix:fluid_model} where the execution time $t_i$'s are constants. In fact, $\btau$ can be explicitly computed by,
\[\btau(\bmu)=\left(\max\left\{t_i,\frac{\mu_i}{b_i}\right\}\right)_{i\in A}.\]
Without loss of generality, we make the following assumption.
\begin{assumption}\label{assum:constant_action_time}
  Assume the action execution time $t_i$ is a constant  for any $i\in A$. Let $V_i$ be defined as,
\[V_i=\frac{r_i}{e^{\beta t_i}-1}.\]
We assume that $V_1> V_2>\cdots>V_n$.  
\end{assumption}

As shown in Fig. \ref{fig:example}, there are $n$ actions with constant execution times. We define $A^\dagger$ as the set of actions each of which is taken by a positive mass of agents at equilibrium.   $A^\dagger$ is always a prefix of the action indices in decreasing order of $V_i$.

As the mass $m$ increases, the size of $A^\dagger$ also increases. If the total mass is small, 
all agents will opt for action 1. As $m$ increases, a queue will form for action 1  increasing its sojourn time. If the mass of agents continues to increase, after a certain point the waiting for action 1 becomes large enough so that the new mass of agents will start choosing action 2 because its rate of revenue becomes equal to action 1, and so on. The detailed equilibrium results are as follows.

\begin{figure}
    \centering
  \includegraphics[width=0.48\textwidth]{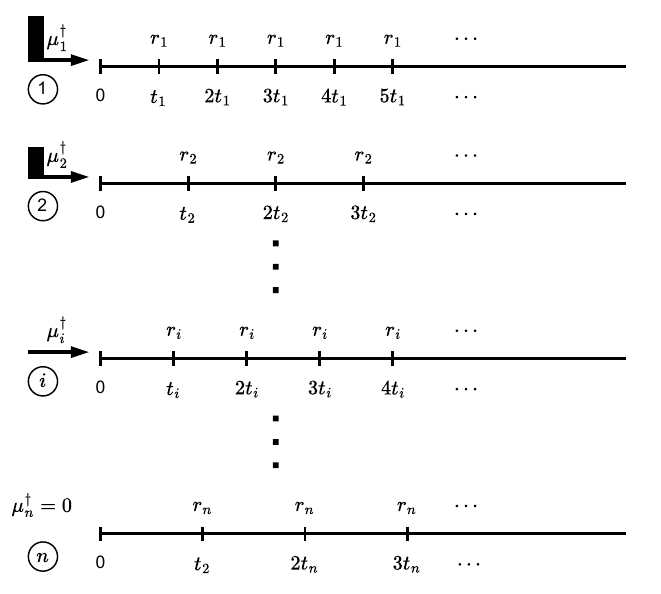}
    \caption{$n$ actions with constant execution times. The figure displays mass distribution $\boldsymbol{\mu}^\dagger$ at equilibrium. If Assumption \ref{assum:constant_action_time} holds, the agents will opt for the first $i$ actions in the equilibrium, i.e., only $\mu^\dagger_1,\cdots,\mu^\dagger_i$ are positive, and  queues must have formed for actions $1,\ldots,i-1$, and possibly for $i$.   The vertical thick line on the left shows the size of the queue formed for each action. Actions with 
    higher values of $V_i$ are chosen first, creating longer queues that reduce their attractiveness in terms of reward rates}.
    \label{fig:example}
\end{figure}

Let $f_i$ denote the reciprocal of $V_i$, i.e.,
\[f_i=\frac{e^{\beta t_i}-1}{r_i}.\]
Then, $f_1<f_2<\cdots<f_n$, since we assume that $V_1> V_2>\cdots>V_n$.
Since the action execution times are constants, we only need to know the mass rates and waiting times at the equilibrium to obtain the corresponding mass distribution and strategies for the agents.  Under Assumption \ref{assum:constant_action_time}, the equilibrium results are given as follows. 

\begin{enumerate}
    \item[] Case (1): $0<\beta m\le b_1\log(r_1f_2+1) $. The agents all choose action 1. The mass rates and waiting times are given by
    \begin{align*}
        &x_1=\min\left\{b_1,\frac{m}{t_1}\right\},\>x_k=0,\>\forall\> 2\le k\le n;\\
        &w_1=\max\left\{0,\frac{m}{b}-t_1\right\},\>w_k=0,\>\forall\> 2\le k\le n.
    \end{align*}
    \item[] Case (i) ($2\le i\le n-1$): $\sum_{1\le k\le i-1}b_{k}\log(r_{k}f_i+1)<\beta m\le \sum_{1\le k\le i}b_k\log(r_k f_{i+1}+1) $. The set of actions chosen by the agents at the equilibrium is $A^\dagger=\{1,2,\cdots,i\}$. 
     If 
\[\sum_{1\le k\le i-1}b_{k}\log(r_{k}f_i+1)<\beta m\le\sum_{1\le k\le i-1}b_{k}\log(r_{k}f_i+1)+\beta b_it_i\,,\]
we have at the equilibrium
\[x_k=
\begin{cases}
b_k,&1\le k\le i-1,\\
\tfrac{1}{t_i}(m-\frac{1}{\beta}\sum_{1\le j\le i-1}b_{j}\log(r_{j}f_i+1)),&k=i,\\
0,&i+1\le k\le n,\\
\end{cases}
\]
\[w_k=
\begin{cases}
\tfrac{1}{\beta}\log(r_kf_i+1)-t_k,&1\le k\le i-1,\\
0,&i\le k\le n.\\
\end{cases}
\]
 If 
\[\sum_{1\le k\le i-1}b_{k}\log(r_{k}f_i+1)+\beta b_it_i<\beta m\le\sum_{1\le k\le i}b_k\log(r_k f_{i+1}+1) ,\]
we have at the equilibrium
\[x_k=
\begin{cases}
b_k,&1\le k\le i,\\
0,&i+1\le k\le n,
\end{cases}
\]
\[w_k=
\begin{cases}
\tfrac{1}{\beta}\log(r_k\psi_i+1)-t_k,&1\le k\le i,\\
0,&i+1\le k\le n,
\end{cases}
\]
where $\psi_i$ satisfies the following equation,
\[\sum_{1\le k\le i}b_k\log(r_k \psi_i+1)=\beta m.\]
\item[] Case (n): $\beta m> \sum_{1\le k\le n-1}b_k\log(r_k f_{n}+1) $. The set of actions chosen by the agents at the equilibrium is $A^\dagger=\{1,2,\cdots,n\}$.  If 
\[ \sum_{1\le k\le n-1}b_k\log(r_k f_{n}+1)<\beta m\le\sum_{1\le k\le n-1}b_k\log(r_k f_{n}+1)+\beta b_nt_n ,\]
we have at the equilibrium
\[x_k=
\begin{cases}
b_k,&1\le k\le n-1,\\
\tfrac{1}{t_n}(m-\frac{1}{\beta}\sum_{1\le j\le n-1}b_{j}\log(r_{j}f_n+1)),&k=n,
\end{cases}
\]

\[w_k=
\begin{cases}
\tfrac{1}{\beta}\log(r_kf_n+1)-t_k,&1\le k\le n-1,\\
0,&k= n.\\
\end{cases}
\]
 If 
\[ \beta m\ge\sum_{1\le k\le n-1}b_k\log(r_k f_{n}+1)+\beta b_nt_n ,\]
we have at the equilibrium:
$x_k=
b_k,\>\forall\>1\le k\le n$,
and
\[w_k=
\tfrac{1}{\beta}\log(r_k\psi_n+1)-t_k,1\le k\le n,
\]
where $\psi_n$ satisfies the following equation,
\[\sum_{1\le k\le n}b_k\log(r_k \psi_n+1)=\beta m.\]
\end{enumerate}

 \section{Proof of Proposition \ref{prop:poa_m+1}}\label{appendix:proof_prop:poa_m+1}

A useful fact which is needed in the proof is that $g_\beta(x):=\frac{e^{\beta x}-1}{\beta x}$ is an increasing function of $x$ and $g_\beta(x)\rightarrow 1$ as  $x\rightarrow0$. Utilizing this fact and Assumption \ref{assum_sojourntime}, we obtain the following lemma, which will be used in the proof.
\begin{lemma}
\label{lem:M_increasing}
$\chi(\mathcal{G}(m,A,\boldsymbol{r},\boldsymbol{\tau},\beta))$ is nondecreasing in $m$.
\end{lemma}
\begin{proof}
For any $m, m',\boldsymbol{\mu} \in \mathcal{U}_m(A)$ with $m' > m$, define $\boldsymbol{\mu}' \in \mathcal{U}_{m'}(A)$ by
$\mu_i' = \frac{m'}{m} \mu_i$ for all $i\in A$.
We claim that $\tau_i(\boldsymbol{\mu}) \leq \tau_i(\boldsymbol{\mu}')$ for all $i$. If not, then
$\tau_i(\boldsymbol{\mu}') < \tau_i(\boldsymbol{\mu})$ for some $i$, which implies $x_i(\boldsymbol{\mu}') \leq x_i(\boldsymbol{\mu})$ by Assumption~\ref{assum_sojourntime}. But then by Little's law,
$\mu_i' = x_i(\boldsymbol{\mu}') \tau_i(\boldsymbol{\mu}') \leq x_i(\boldsymbol{\mu}) \tau_i(\boldsymbol{\mu}) = \mu_i$ which is a contradiction since $\mu_i' > \mu_i$ for all $i\in A$.

Thus, 
\[
\frac{e^{\beta \tau_i(\boldsymbol{\mu})}-1}{\beta \tau_i(\boldsymbol{\mu})} \leq
\frac{e^{\beta \tau_i(\boldsymbol{\mu}')}-1}{\beta \tau_i(\boldsymbol{\mu}')},
\]
for all $i$, which implies $\chi(\mathcal{G}(m,A,\boldsymbol{r},\boldsymbol{\tau},\beta)) \leq \chi(\mathcal{G}(m',A,\boldsymbol{r},\boldsymbol{\tau},\beta))$.
\end{proof}

Let $\boldsymbol{\mu}^\dagger=(\mu^\dagger_i)_{i\in A}$ be the mass distribution of a stationary equilibrium and $\boldsymbol{\mu}^*=(\boldsymbol{\mu}^*_i)_{i\in A}$
be a socially optimal mass distribution. 

Define the set of actions which are used at least as frequently in the equilibrium than in the social optimum,
\[X=\left\{i\in A\mid x_{i}(\boldsymbol{\mu}^\dagger) \geq {x}_{i}(\boldsymbol{\mu}^*) \right\}.\]

First, notice that
\begin{equation}
\label{eq:ins1}
\sum_{i\in X}r_{i}x_i(\boldsymbol{\mu}^*)
\le\sum_{i\in X}r_{i}x_i(\boldsymbol{\mu}^\dagger)
\le\sum_{i\in A}r_{i}x_i(\boldsymbol{\mu}^\dagger)
=SW(\boldsymbol{\mu}^\dagger).
\end{equation}

For $i\notin X$,
\begin{equation}
\label{eq:pf_tmp0}
\sum_{i\notin X} r_i x_i(\boldsymbol{\mu}^*)
  \leq \sum_{i\notin X} \frac{r_i}{e^{\beta \tau_i(\boldsymbol{\mu}^*)}-1} \frac{e^{\beta \tau_i(\boldsymbol{\mu}^*)}-1}{\tau_i(\boldsymbol{\mu}^*)}\mu_i^*,
\end{equation}
and
\begin{equation}
\label{eq:pf_tmp1}
   \frac{r_i}{e^{\beta \tau_i(\boldsymbol{\mu}^*)}-1} \leq \frac{r_i}{e^{\beta \tau_i(\boldsymbol{\mu}^\dagger)}-1} 
   \leq \frac{r_l}{e^{\beta \tau_l(\boldsymbol{\mu}^\dagger)}-1},
\end{equation}
where the first inequality follows from Assumption~\ref{assum_sojourntime} for all $i\notin X$, and the
second is due to the equilibrium condition~\eqref{thm:stationary_equilibrium_eq} in Theorem~\ref{thm:stationary_equilibrium} for any $l\in A^\dagger$.
Also, 
\begin{equation}
\label{eq:pf_tmp2}
    \frac{e^{\beta \tau_i(\boldsymbol{\mu}^*)}-1}{\tau_i(\boldsymbol{\mu}^*)} \leq \beta M\,,
\end{equation}
by the definition of $\chi(\mathcal{G}(m,A,\boldsymbol{r},\boldsymbol{\tau},\beta))$ and Lemma~\ref{lem:M_increasing}.
Using~\eqref{eq:pf_tmp1} and~\eqref{eq:pf_tmp2} in~\eqref{eq:pf_tmp0} yields
\begin{equation}
\label{eq:pf_tmp3}
    \sum_{i\notin X} r_i x_i(\boldsymbol{\mu}^*) \leq
       \sum_{i\notin X} \frac{r_l}{e^{\beta \tau_l(\boldsymbol{\mu}^\dagger)}-1} \beta M \mu_i^* \leq
       \frac{r_l}{\tau_l(\boldsymbol{\mu}^\dagger)} M m
\end{equation}
for every $l\in A^\dagger$. Choosing a $l\in A^\dagger$ with the smallest ratio $r_l / \tau_l(\boldsymbol{\mu}^\dagger)$ gives 
the further bound
\begin{equation}
\label{eq:pf_tmp4}
    \frac{r_l}{\tau_l(\boldsymbol{\mu}^\dagger)} M m \leq M\sum_{i\in A} \frac{r_i}{\tau_i(\boldsymbol{\mu}^\dagger)} \mu^\dagger_i
    = M \ \text{SW}(\boldsymbol{\mu}^\dagger)
\end{equation}
Combining~\eqref{eq:ins1}, \eqref{eq:pf_tmp3} and \eqref{eq:pf_tmp4} yields
\begin{equation*}
    \text{SW}(\boldsymbol{\mu}^*) = \sum_{i\in X} r_i x_i(\boldsymbol{\mu}^*) + \sum_{i\notin X} r_i x_i(\boldsymbol{\mu}^*) 
    \leq (M + 1) \ \text{SW}(\boldsymbol{\mu}^\dagger)\,.
\end{equation*}
Therefore, $\text{PoA}\le M+1$.

For the lower bound, we give a sequence of instances  for which the ratio inside the supremum in \eqref{eq:PoA_definition} approaches $M+1$. Consider the constant action execution time model with $n=2$ and execution times $t_1$, $t_2$. Assume the resource supply rate $b_2$ is infinite, that is, $w_2$ is always equal to 0.  Let
\[M=\frac{e^{\beta t_2}-1}{\beta t_2}\ge \frac{e^{\beta t_1}-1}{\beta t_1}.\]
Let
\[\frac{r_1}{e^{\beta(t_1+w_1)}-1}=\frac{r_2}{e^{\beta t_2}-1},\]
and
\[b_1(t_1+w_1)=m.\]
Then, at the equilibrium all agents choose action 1. The social welfare will only increase if the 
mass $b_1 w_1 = m - b_1 t_1$ of agents waiting for action $1$ take action $2$ instead.
This gives the lower bound,
\begin{align*}
    \text{PoA}&\ge\frac{r_1b_1+r_2\frac{m-b_1t_1}{t_2}}{r_1b_1}=1+\frac{r_2(m-b_1t_1)}{t_2r_1b_1}\\
    &=1-\frac{\frac{r_2}{t_2}}{\frac{r_1}{t_1}}+\frac{\frac{r_2}{t_2}}{\frac{r_1}{t_1+w_1}}\frac{\frac{r_1}{e^{\beta(t_1+w_1)}-1}}{\frac{r_2}{e^{\beta t_2}-1}}\\
    &=1-\frac{\frac{e^{\beta t_2}-1}{\beta t_2}}{\frac{e^{\beta(t_1+w_1)}-1}{\beta t_1}}+\frac{\frac{e^{\beta t_2}-1}{\beta t_2}}{\frac{e^{\beta(t_1+w_1)}-1}{\beta (t_1+w_1)}}=1+\frac{\frac{e^{\beta t_2}-1}{\beta t_2}}{\frac{e^{\beta(t_1+w_1)}-1}{\beta (t_1+w_1)}\frac{t_1+w_1}{w_1}}\>.
\end{align*}
We let $w_1$ approach zero and $t_1=o(w_1)$. This can be done by letting $m$ be fixed, $b_1=\frac{m}{t_1+\sqrt{t_1}}$, and $t_1\rightarrow0$. Alternatively, we can let $b_1$ be fixed, $m=b_1(t_1+\sqrt{t_1})$,   and $t_1\rightarrow0$.  Eitherway, we have  $t_1+w_1\rightarrow0$ and  $\frac{t_1+w_1}{w_1}\rightarrow1$. Then, it follows that,
\[\frac{e^{\beta(t_1+w_1)}-1}{\beta(t_1+w_1)}\frac{t_1+w_1}{w_1}\rightarrow1.\]
Therefore, 
\[\text{PoA}\ge 1+\frac{\frac{e^{\beta t_2}-1}{\beta t_2}}{\frac{e^{\beta(t_1+w_1)}-1}{\beta (t_1+w_1)}\frac{t_1+w_1}{w_1}}\rightarrow 1+ \frac{e^{\beta t_2}-1}{\beta t_2}=M+1.\]

\section{Proof of Corollary \ref{coro:infinite_poa}}\label{appendix:proof_coro:infinite_poa}
Here we allow the parameter $M$ to vary.   We will give a sequence of instances with PoA approaching infinity. Similar to the sequence of instances given in the proof of Proposition \ref{prop:poa_m+1}, we still
   assume $n=2$ and let
\[M=\frac{e^{\beta t_2}-1}{\beta t_2}\ge \frac{e^{\beta t_1}-1}{\beta t_1}.\]
Similarly, let
\[\frac{r_1}{e^{\beta(t_1+w_1)}-1}=\frac{r_2}{e^{\beta t_2}-1},\]
and
\[b_1(t_1+w_1)=m.\]
Then we have at the equilibrium all agents choose action 1. The PoA is computed similarly as follows,
\begin{align*}
    \text{PoA}&\ge\frac{r_1b_1+r_2\frac{m-b_1t_1}{t_2}}{r_1b_1}=1+\frac{r_2(m-b_1t_1)}{t_2r_1b_1}\\
    &=1-\frac{\frac{r_2}{t_2}}{\frac{r_1}{t_1}}+\frac{\frac{r_2}{t_2}}{\frac{r_1}{t_1+w_1}}\frac{\frac{r_1}{e^{\beta(t_1+w_1)}-1}}{\frac{r_2}{e^{\beta t_2}-1}}\\
    &=1-\frac{\frac{e^{\beta t_2}-1}{\beta t_2}}{\frac{e^{\beta(t_1+w_1)}-1}{\beta t_1}}+\frac{\frac{e^{\beta t_2}-1}{\beta t_2}}{\frac{e^{\beta(t_1+w_1)}-1}{\beta(t_1+w_1)}}=1+\frac{\frac{e^{\beta t_2}-1}{\beta t_2}}{\frac{e^{\beta(t_1+w_1)}-1}{\beta w_1}}.
\end{align*}
   
   Consider the following setup of parameters. 
Let $m$, $b_1$, and $t_1$ be fixed. Then, $w_1=\frac{m}{b_1}-t_1$ is also fixed. Let $t_2\rightarrow\infty$,  it follows that
\[M=\frac{e^{\beta t_2}-1}{\beta t_2}\rightarrow\infty,\]
since $\frac{e^{\beta t_2}-1}{\beta t_2}\rightarrow\infty$ as $t_2\rightarrow\infty$. 
Thus,  as we have  claimed,
\[\text{PoA}\ge1+\frac{\frac{e^{\beta t_2}-1}{\beta t_2}}{\frac{e^{\beta(t_1+w_1)}-1}{\beta w_1}}\rightarrow\infty.\]

\section{Proof of Proposition \ref{prop:switchingornot}}\label{appendix:proof_prop:switchingornot}

We first prove Proposition \ref{prop:switchingornot}.1 by giving a game instance $\mathcal{G}$ where a  nonstatioanry (switching) strategy can be more profitable than any deterministic stationary strategy of always choosing a single action.
Consider the constant action execution time model with $n=2$. Assume $b_1=b_2=\infty$.
 Since $\alpha>1$, the series
$\sum\limits_{n=1}^\infty n^{-\alpha}$
converges to $\zeta(\alpha)$, the Riemann zeta function evaluated at $\alpha$. For every $\epsilon>0$, there exists an integer $N$ such that 
\[0<N^{-\alpha}<\zeta(\alpha)- \sum\limits_{n=1}^{N-1}n^{-\alpha}<\epsilon.\]
Assume $t_2=Nt_1$ and $r_2=N^{\alpha}r_1-r_1$. Then,
\[
V(1, \boldsymbol{\tau})=r_1{t_1^{-\alpha}}\zeta(\alpha)
>\left(1-{N^{-\alpha}}\right)r_1{t_1^{-\alpha}}\zeta(\alpha)
=r_2{t_2^{-\alpha}}\zeta(\alpha) = V(2, \boldsymbol{\tau}).\]
Thus, the stationary strategy of always choosing action 1 is the optimal among  stationary strategies.

The following nonstatioanry (switching) strategy can result in a strictly better payoff: the agent chooses action 1 until time $T=Nt_1$ and chooses action 2 thereafter. The total utility under this strategy is
\begin{align*}
&\sum\limits_{n=1}^N r_1(nt_1)^{-\alpha}+\sum\limits_{n=2}^\infty r_2(nt_2)^{-\alpha}\\
    &\quad=r_1t_1^{-\alpha}\sum\limits_{n=1}^N n^{-\alpha}+(1-N^{-\alpha})r_1t_1^{-\alpha}\sum\limits_{n=2}^\infty n^{-\alpha}\\
    &\quad>r_1t_1^{-\alpha}\big[\zeta(\alpha)-\epsilon+(1-\epsilon)(\zeta(\alpha)-1)\big]\\
    &\quad=r_1t_1^{-\alpha}\big[(2-\epsilon)\zeta(\alpha)-1\big],
    \end{align*}
which is strictly larger than $V(1, \boldsymbol{\tau}) = r_1t_1^{-\alpha}\zeta(\alpha)$ if 
$\epsilon<1-{\zeta(\alpha)}^{-1}$.

We next prove Proposition \ref{prop:switchingornot}.2. 
We will only give a proof for the case when there are two actions. When there are more actions, the proof is similar.   Given $\boldsymbol{\mu}$, the sojourn time associated with the two actions $\tau_1$ and $\tau_2$ are then determined. Without loss of generality, we assume
\[\frac{r_1}{\tau_1(\bmu)}\ge\frac{r_2}{\tau_2(\bmu)}.\]
We will show that the deterministic stationary strategy of always choosing action 1 is the best response against $\bmu$.

We first define a sequence of positive constants: $\delta_k=\tau_1$ or $\tau_2$ for $k\ge1$.
Any switching strategy can be represented by the sequence $(\delta_k)_{k\in\mathbb{N}}$ where $\delta_k$ can be viewed as the $k$th choice of the strategy. Specifically, for the $k$th choice,  $\delta_k=\tau_1$ indicates the strategy under consideration chooses action 1 and $\delta_k=\tau_2$  indicates the strategy under consideration chooses action 2. Given $(\delta_k)_{k\in\mathbb{N}}$, define $\eta_k$, $k\in\mathbb{N}$,  as follows
\[\eta_k=\begin{cases}
    r_1,&\delta_k=\tau_1\\
    r_2,&\delta_k=\tau_2
\end{cases}\]
Let \[\theta_n=\sum_{k=1}^n\delta_k.\]
For any large enough $T$, there is a $N$ such that
\[\theta_N\le T<\theta_{N+1}.\]
To prove the desired result, we need the following facts which follow from 
elementary facts.
When $0\leq\alpha<1$,
\begin{equation}\label{ineq_1-alpha}
    \sum_{k=1}^{N}\delta_k \theta_k^{-\alpha}\sim T^{-\alpha+1},
\end{equation}
and when $\alpha=1$,
\begin{equation}\label{ineq_log}
   \sum_{k=1}^{N}\delta_k\theta_k^{-1}\sim \log T.
\end{equation}

Now assume $0\leq\alpha<1$; the discussion for $\alpha=1$ is similar. For any switching strategy represented by $(\delta_k)_{k\in\mathbb{N}}$ and $(\eta_k)_{k\in\mathbb{N}}$, the total utility is
\[\sum_{k=1}^{N}\eta_k \theta_k^{-\alpha}.\]
Note that $\frac{\eta_k}{\delta_k}=\frac{r_1}{\tau_1}$ or $\frac{r_2}{\tau_2}$. Hence,
\[\sum_{k=1}^{N}\eta_k \theta_k^{-\alpha}=\sum_{k=1}^{N}\frac{\eta_k}{\delta_k}\delta_k \theta_k^{-\alpha}\le \frac{r_1}{\tau_1}\sum_{k=1}^{N}\delta_k \theta_k^{-\alpha}\sim\frac{r_1}{\tau_1}T^{-\alpha+1}.\]
The reward of always choosing action 1 is $V_1(T)\sim\frac{r_1}{\tau_1}T^{-\alpha+1}$ by \eqref{eq:reward_VT}. Thus, always choosing action 1 cannot be asymptotically worse than any nonstationary  (switching) strategy. 

Now if there are $n>2$ actions and multiple actions have the largest reward rate, i.e.,  
\[\frac{r_1}{\tau_1(\bmu)}=\cdots=\frac{r_j}{\tau_j(\bmu)}>\frac{r_{j+1}}{\tau_{j+1}(\bmu)}\geq \cdots\geq \frac{r_{n}}{\tau_n(\bmu)}.\]
From the proof for the two-action case, it is clear that any randomized stationary strategy which randomizes over the set of actions $\{1,\cdots,j\}$ and any deterministic stationary strategy which always chooses one action from $\{1,\cdots,j\}$ are  best responses against the agent mass distribution $\bmu$.


\section{The Constant Execution Time Model under Power Law Discounting with $0\le \alpha\le 1$}\label{appendix:se_example_2}
We will show how to compute the mean field equilibrium for the constant action execution time model. We assume
 \[\frac{r_1}{t_1}>\frac{r_2}{t_2}>\cdots>\frac{r_n}{t_n}.\]
Then, the equilibrium results are similar to Appendix \ref{appendix:se_example_1} as shown in Fig. \ref{fig:example}.  When the mass of agents grows larger, so does the size of $A^\dagger$, i.e.,  the set of actions each of which is chosen by a positive mass of agents at equilibrium. The detailed equilibrium results are as follows.

 Let ${w_k^l}$ be the solution to the following equation
 \[\frac{r_k}{t_k+{w}_k^l}=\frac{r_{l}}{t_{l}},\quad\forall k=1,\cdots,n-1,\ k<l\le n,\]
or equivalently,
\[{w}_k^l=\frac{t_{l}r_k}{r_{l}}-t_k,\quad\forall k=1,\cdots,n-1,\ k<l\le n.\]
The equilibrium results are summarized in the following.
\begin{enumerate}
    \item[]Case (1): $0\le m\le b_1(t_1+w_1^2)$. The agents all choose action 1. The equilibrium is characterized by
\[x_1=\min\left(\frac{m}{t_1},b_1\right),\quad w_1=\max\left(0,\frac{m}{b_1}-t_1\right).\]
     \item[]Case (i): $\sum\limits_{1\le k\le i-1}b_k(t_k+w_k^i)<m\le \sum\limits_{1\le k\le i}b_k(t_k+w_k^{i+1})$, where $2\le i\le n-1$. The set of actions chosen by the agents at the equilibrium is $A^\dagger=\{1,\cdots,i\}$. When $\sum\limits_{1\le k\le i-1}b_k(t_k+w_k^i)<m\le \sum\limits_{1\le k\le i-1}b_k(t_k+w_k^i)+b_it_i$, the equilibrium mass rates are given by
\[x_k=
\begin{cases}
b_k,&1\le k\le i-1,\\
\displaystyle \frac{m-\sum\limits_{1\le k\le i-1}b_k(t_k+w_k^i)}{t_i},&k=i,\\
0,&i+1\le k\le n,\\
\end{cases}
\]
and waiting times are given  by
\[w_k=
\begin{cases}
w_k^{i},&1\le k\le i-1,\\
0,&i\le k\le n.
\end{cases}
\]
When $\sum\limits_{1\le k\le i-1}b_k(t_k+w_k^i)+b_it_i<m\le \sum\limits_{1\le k\le i}b_k(t_k+w_k^{i+1})$, the equilibrium mass rated are given by
\[x_k=
\begin{cases}
b_k,&1\le k\le i,\\
0,&i+1\le k\le n,
\end{cases}
\]
and waiting times are given  by
\[w_k=
\begin{cases}
\frac{m}{ \sum\limits_{k=1}^ir_kb_k}
r_k-t_k,&1\le k\le i,\\
0,&i+1\le k\le n.
\end{cases}
\]
    \item[]Case (n): $m>\sum\limits_{1\le k\le n-1}b_k(t_k+w_k^{n})$. The set of actions chosen by the agents at the equilibrium is $A^\dagger=\{1,\cdots,n\}$.
When $\sum\limits_{1\le k\le n-1}b_k(t_k+w_k^n)<m\le \sum\limits_{1\le k\le n-1}b_k(t_k+w_k^n)+b_nt_n$,  the equilibrium mass rates are given by
\[x_k=
\begin{cases}
b_k,&1\le k\le n-1,\\
\frac{m-\sum\limits_{1\le k\le n-1}b_k(t_k+w_k^n)}{t_n},&k=n,
\end{cases}
\]
and the waiting times are given  by
\[w_k=
\begin{cases}
w_k^{n},&1\le k\le n-1,\\
0,&k=n.
\end{cases}
\]
When $m>\sum\limits_{1\le k\le n-1}b_k(t_k+w_k^n)+b_nt_n$, the equilibrium mass rates and waiting times are given by
\[x_k=b_k,\quad w_k=\frac{m}{ \sum\limits_{k=1}^nr_kb_k}r_k-t_k,\quad\forall\> 1\le k\le n.
\]
\end{enumerate}

\section{Regime of High Discounting ($\alpha>1$) and Proof of Proposition \ref{coro:se_nonexistence}}\label{appendix:proof_coro:se_nonexistence}

If $\alpha>1$, it is difficult to predict the strategic behavior of the agents since their preferences may shift from one action to another depending on when this action is executed.
Even if a stationary equilibrium does not exist,  it is always better for an agent to follow a deterministic strategy after some initial time period.
The following technical result will  allow us to show that the PoA is 2 whenever a stationary equilibrium exists.  

\begin{proposition}\label{prop:alpha_largerthan1}
   Consider the game $\mathcal{G}(m,A,\boldsymbol{r},\boldsymbol{\tau},\alpha)$ with $\alpha>1$ and any agent mass distribution $\boldsymbol{\mu}\in\mathcal{U}_m(A)$ for which
    \[\frac{r_1}{\tau_1(\boldsymbol{\mu})}>\frac{r_2}{\tau_2(\boldsymbol{\mu})}\ge\cdots
\ge\frac{r_n}{\tau_n(\boldsymbol{\mu})}\,.\]
Then,
     the (deterministic) strategy of always choosing the action associated with the largest reward rate is the unique best response against $\bmu$ after some initial time period.
     Furthermore, if $\boldsymbol{\mu}$ satisfies
    \[\frac{r_1}{\tau_1(\boldsymbol{\mu})}=\frac{r_2}{\tau_2(\boldsymbol{\mu})}=\cdots
=\frac{r_n}{\tau_n(\boldsymbol{\mu})},\]
and 
\[{\tau_1(\boldsymbol{\mu})}<{\tau_2(\boldsymbol{\mu})}\leq\cdots\leq{\tau_n(\boldsymbol{\mu})},\]
the strategy of always choosing the action associated with the smallest sojourn time is the unique best response against $\bmu$  after some initial time period.
\end{proposition}

\begin{proof}
 Note that there is a large enough $T$ such that if $t\ge T$, we have
\begin{equation}
    \frac{r_1}{\tau_1}(t+\tau_1)^{-\alpha}>\frac{r_i}{\tau_i}t^{-\alpha},\quad\forall\>i\in\{2,\dots,n\}.
    \label{1:decreasing}
\end{equation}

We will see that after time $T$, choosing action 1 each time the agent faces a choice will give the best possible reward. 

\begin{figure}
    \centering
    \includegraphics[width=0.48\textwidth]{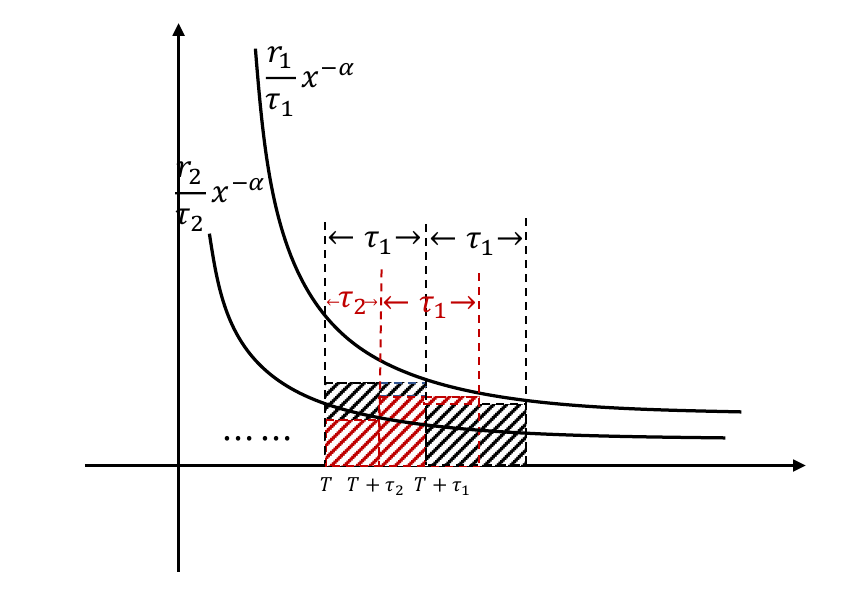}
    \caption{Viewing a discrete reward process as a continuous reward process from the corresponding step functions below the curves $y=\frac{r_1}{\tau_1}x^{-\alpha}$ and  $y=\frac{r_2}{\tau_2}x^{-\alpha}$. This process is similar to approximating the integral of the function $y=\frac{r_i}{\tau_i}x^{-\alpha}$ by the sum of areas of the rectangles below the curve. }
    \label{fig:dominant}
\end{figure}

Without loss of generality, we assume that at time $T$ the agent just completes some action and faces a new choice. Choosing action 1 thereafter will result in the following reward,
\[\sum_{k=1}^\infty r_1(T+k\tau_1)^{-\alpha}=\sum_{k=1}^\infty \tau_1\frac{r_1}{\tau_1}(T+k\tau_1)^{-\alpha}.\]
We observe that this is exactly the integral of a step function from $T$ to infinity and the step function is below the curve $y=\frac{r_1}{\tau_1}x^{-\alpha}$ as shown in Fig~\ref{fig:dominant}.  We denote the step function by $S_{T}^1$, where $S_T^1=\frac{r_1}{\tau_1}(T+k\tau_1)^{-\alpha}$ in the interval $\big(T+(k-1)\tau_1,T+k\tau_1\big]$ for any $k\in\mathbb{N}$.
Although the reward process in this scenario is discrete, considering the total reward as an integral of a step function allows us to treat it as a continuous process. This perspective enables us to discuss the reward accumulated over a specific time period, even if multiple actions were completed within that time period or some actions remained incomplete at the end of the period. 

Consider an arbitrary strategy (stationary or nonstationary) that the agent adopts after time $T$.
The total reward generated since time $T$ under this strategy is also an integral of some step function $S^\prime_T$ below the curve $y=\frac{r_1}{\tau_1}x^{-\alpha}$. The strategy used by the agent can be inferred by the step function. Hence, with a slight abuse of notation, we use the same notation of the step function  to denote the strategy under consideration (instead of $\boldsymbol{\sigma}$, which is a sequence of probability measures as in Section \ref{sec:model}).
Assume under the strategy $S^\prime_T$, the agent executes action $i\not=1$ starting from time $t\ge T$ until the time $t+\tau_i$ where $t+\tau_i$ is in the interval $[T+(k_t-1)\tau_1,T+k_t\tau_1]$. Then, $\forall\>x\in[T+(k_t-1)\tau_1,T+k_t\tau_1]$,
\[
    S^1_T(x)\ge \frac{r_1}{\tau_1}(T+k_t\tau_1)^{-\alpha}>  \frac{r_i}{\tau_i}(T+(k_t-1)\tau_1)^{-\alpha}\ge \frac{r_i}{\tau_i}(t+\tau_i)^{-\alpha},
\]
where the second inequality follows from \eqref{1:decreasing}.

Hence, executing action $i$ from time $t$ to time $t+\tau_i$ generates a reward strictly less than that of the always-choose-action-1 strategy (i.e., $S^1_T$) at the same time period. This holds true for any time period after time $T$ when the agent chooses an action other than action 1.  

Assume under strategy $S^\prime_T$,  the agent executes action 1 from time $t> T$ to time $t+\tau_1$ where $t$ is in the interval $(T+(k_t-1)\tau_1,T+k_t\tau_1)$.  Note that if $t=T+k_t\tau_1$ for some integer $k_t$, the reward generated from $t$ to $t+\tau_1$ will be the same for both strategies  $S_T^\prime$ and $S_T^1$.  During the time period $[t,t+\tau_1]$, the reward generated by $S^\prime_T$, which is denoted by $R([t,t+\tau_1],S^\prime_T)$, is
\[R([t,t+\tau_1],S^\prime_T)=\frac{r_1}{\tau_1}(t+\tau_1)^{-\alpha}\tau_1\,.\]
During the same time period, the accumulated reward under strategy $S^1_T$, which is denoted by $R([t,t+\tau_1],S^1_T)$, is
\[
\begin{aligned}
    &R([t,t+\tau_1],S^1_T)\\
  &\quad=\frac{r_1}{\tau_1}(T+k_t\tau_1)^{-\alpha}(T+k_t\tau_1-t)\\
    &\quad\quad+\frac{r_1}{\tau_1}(T+(k_t+1)\tau_1)^{-\alpha}(t+\tau_1-T-k_t\tau_1).
\end{aligned}
\]
In fact, under the strategy $S^1_T$, the agent only receives the reward once exactly at time $T+k_t\tau_1$. But as we mentioned before, we consider an equivalent  continuous reward process which allows for the above calculation. We can show the following inequality
\[R([t,t+\tau_1],S^1_T)>R([t,t+\tau_1],S^\prime_T)\]
by the convexity of the function $x^{-\alpha}$.
In fact, this follows from the inequality
\begin{equation}\label{2:convex}
    \frac{(T+k_t\tau_1)^{-\alpha}-(t+\tau_1)^{-\alpha}}{t+\tau_1-T-k_t\tau_1}>\frac{(t+\tau_1)^{-\alpha}-(T+(k_t+1)\tau_1)^{-\alpha}}{T+k_t\tau_1-t}.
\end{equation}

In summary, for an arbitrary strategy $S_T^\prime$, if it executes action $i\not=1$ after time $T$, it generates a reward strictly smaller than the always-choose-action-1 strategy during the same time period. If it executes action $i=1$, the reward generated cannot be larger than the always-choose-action-1 strategy during the same time period.
Thus, the always-choose-action-1 strategy is the agent's best response after  time $T$ for any mass distribution over the action space $A$ such that 
\[\frac{r_1}{\tau_1(\boldsymbol{\mu})}>\frac{r_2}{\tau_2(\boldsymbol{\mu})}\geq\cdots\geq\frac{r_n}{\tau_n(\boldsymbol{\mu})}.\]

Now suppose the stationary mass distribution $\boldsymbol{\mu}$ is such that
    \[\frac{r_1}{\tau_1(\boldsymbol{\mu})}=\frac{r_2}{\tau_2(\boldsymbol{\mu})}=\ldots
=\frac{r_n}{\tau_n(\boldsymbol{\mu})},\] 
and  \[\tau_1(\boldsymbol{\mu})<\tau_i(\boldsymbol{\mu}),\quad\forall\,i\not=1.\]
we prove that always-choose-action-1 is the unique best response against $\bmu$ regardless of the time.  We only need to show that 
for any $t\ge 0$, if an arbitrary strategy $S^\prime$ involves executing action $i\not=1$ from time $t$ to time $t+\tau_i$, it generates a reward $R([t,t+\tau_i],S^\prime)$ strictly smaller than the always-choose-action-1 strategy in the same time period.

First, we note that
any strategy corresponds to a step function below the same curve $y=\frac{r_i}{\tau_i}x^{-\alpha}$ where $i$ is an arbitrary number in  $A$.
Since $\tau_i>\tau_1$, there must be a smallest integer $k_t\in\mathbb{N}$ such that $k_t\tau_1\in(t,t+\tau_i]$. If $k_t\tau_1=t+\tau_i$, it is clear that always-choose-action-1 strategy generates a strictly greater reward than $R([t,t+\tau_i],S^\prime)$ in the time interval $[t,t+\tau_i]$. Now we assume $k_t\tau_1<t+\tau_i$.
Note  that the inequality \eqref{2:convex}  follows from the fact that 
the utility function $u(r,t)$ is convex in time $t$. 
Hence, a similar inequality can be established to prove that for time interval $[t+\tau_i-\tau_1,t+\tau_i]$, and as a result always-choose-action-1 strategy  generates a reward greater than $S^\prime$. For time interval $[t,t+\tau_i-\tau_1]$, it is also clear that always-choose-action-1 strategy generates a reward strictly greater than $S^\prime$.
Thus, always-choose-action-1 is the unique best response given the stationary mass distribution $\boldsymbol{\mu}$.
\end{proof}

 Proposition \ref{prop:alpha_largerthan1} presumes a given agent mass distribution and focuses solely on the agents' strategies as time approaches infinity, which may never manifest in reality. Nevertheless, Proposition \ref{prop:alpha_largerthan1} plays a crucial role in comprehending the behavior of agents if a stationary equilibrium does indeed exist.
According to Proposition \ref{prop:switchingornot}, the best response against certain agent mass distribution  could be non-stationary (time-dependent).  Nevertheless, this simply suggests that the best response map \eqref{bestresponse} may not be well defined and  does not imply the non-existence of a stationary equilibrium. In fact, it is Proposition \ref{prop:alpha_largerthan1} that enables the demonstration of the absence of a stationary equilibrium for some instances of Game $\mathcal{G}$. Now we give the proof of Proposition \ref{coro:se_nonexistence}.

\begin{proof}
Fix an arbitrary $\alpha>1$. Assume a stationary equilibrium exists and w.l.o.g., it is a stationary mass distribution $\boldsymbol{\mu}$ such that, 
\[\frac{r_1}{\tau_1(\boldsymbol{\mu})}=\frac{r_2}{\tau_2(\boldsymbol{\mu})}=\cdots=\frac{r_k}{\tau_k(\boldsymbol{\mu})}>\frac{r_{k+1}}{\tau_{k+1}(\boldsymbol{\mu})}\ge\cdots\ge\frac{r_n}{\tau_n(\boldsymbol{\mu})}.\]
By Definition~\ref{def:stationary_ne}, a stationary optimal strategy exists for the agents. Hence, the  best response against $\bmu$ as in Proposition \ref{prop:alpha_largerthan1}, which is optimal when time goes to infinity, must be optimal from the beginning.  
It follows that given $\bmu$, it must be optimal for the agents to choose the action $i\in\{1,\dots,k\}$ only if $\tau_i(\boldsymbol{\mu})$ is the smallest among all sojourn times.

Consider the constant action execution time model with $n=2$. The reward rates are $\frac{r_1}{t_1}$ and $\frac{r_2}{t_2}$ where $\frac{r_1}{t_1}>\frac{r_2}{t_2}$ and $r_1\not= r_2$. We further assume that the resource supply rate is $b_1$ for action 1 and the agent mass $m$ is large enough. Now consider the following three categories of candidates for a stationary equilibrium.

First, all agents choose action 1. Since $m=b_1(t_1+w_1)$ is large enough and the supply rate $b_1$ is finite,
\[\frac{r_1}{t_1+w_1}<\frac{r_2}{t_2},\]
so this cannot be an equilibrium.

Second, all agents choose action 2. This can not be an equilibrium, since
\[\frac{r_1}{t_1}>\frac{r_2}{t_2}>\frac{r_2}{t_2+w_2}.\]

Third, assume positive mass of agents over both actions. If this is an equilibrium, it must follow that
\[\frac{r_1}{t_1+w_1}=\frac{r_2}{t_2+w_2},\]
and 
\[t_1+w_1=t_2+w_2,\]
which implies $r_1=r_2$. Since $r_1\not=r_2$, this configuration cannot be an equilibrium.

After exploring all potential candidates for a stationary equilibrium, none of them meet the necessary conditions required for a stationary equilibrium as in Definition~\ref{def:stationary_ne}. Consequently, it can be concluded that there is no stationary equilibrium.
\end{proof}

\section{Proof of Proposition \ref{prop:poa=2}}\label{appendix:proof_prop:poa=2}

By Theorem \ref{thm:se_nonexpo}, when $0\le\alpha\le1$, for any  stationary equilibrium $\boldsymbol{\mu}^\dagger$, it is optimal for the agents to choose the action with the largest possible reward rate, i.e.,
$\max_{i\in A}\frac{r_i}{\tau_i(\boldsymbol{\mu}^\dagger)}$. 
 This is also true when $\alpha>1$ and a stationary equilibrium $\boldsymbol{\mu}^\dagger$ exists  by Proposition \ref{prop:alpha_largerthan1}.
Thus, we do not need to treat the two cases of power law discounting separately. We give one argument for both cases as follows. 

Let $\boldsymbol{\mu}^\dagger=(\mu^\dagger_i)_{i\in A}$ be the mass distribution of a stationary equilibrium and $\boldsymbol{\mu}^*=(\boldsymbol{\mu}^*_i)_{i\in A}$
be the socially optimal mass distribution. 



Similar to the proof of Proposition \ref{prop:poa_m+1}, define
 the set of actions which are used at least as frequently in the equilibrium than in the social optimum,
\[X=\left\{i\in A\mid x_{i}(\boldsymbol{\mu}^\dagger) \geq {x}_{i}(\boldsymbol{\mu}^*) \right\}.\]

First, notice that
\[
\sum_{i\in X}r_{i}x_i(\boldsymbol{\mu}^*)
\le\sum_{i\in X}r_{i}x_i(\boldsymbol{\mu}^\dagger)
\le\sum_{i\in A}r_{i}x_i(\boldsymbol{\mu}^\dagger)
=SW(\boldsymbol{\mu}^\dagger).
\]

For those actions $i\notin X$, it follows from the monotone property of $\tau_i$ function (Assumption \ref{assum_sojourntime}) and $x_i(\boldsymbol{\mu}^\dagger)<x_i(\boldsymbol{\mu}^*)$ that,
\begin{equation}
 \tau_i(\boldsymbol{\mu}^*)\ge\tau_i(\boldsymbol{\mu}^\dagger).\label{ineq:cond1}   
\end{equation}
By the equilibrium condition \eqref{eq:se_nonexpo_cond} in Theorem \ref{thm:se_nonexpo}, for any $i,k\in A^\dagger$ and $j\in A/A^\dagger$
\begin{equation}
\frac{r_i}{\tau_i(\boldsymbol{\mu}^\dagger)}=\frac{r_k}{\tau_k(\boldsymbol{\mu}^\dagger)}\geq \frac{r_j}{\tau_j(\boldsymbol{\mu}^\dagger)}.\label{ineq:cond2}  
\end{equation}
Using \eqref{ineq:cond1} and \eqref{ineq:cond2} yields 
\[
\begin{aligned}
    &\sum_{i\notin X}r_{i}x_i(\boldsymbol{\mu}^*)\\
    &\quad=\sum_{i\notin X}\frac{r_{i}}{\tau_{i}({\boldsymbol{\mu}}^*)}\boldsymbol{\mu}^*
_{i}\le
\sum_{i\notin X} \frac{r_{i}}{\tau_{i}(\boldsymbol{\mu}^\dagger)}\boldsymbol{\mu}^*_i 
\le\frac{r_{l}}{\tau_{l}(\boldsymbol{\mu}^\dagger)} \sum_{i\notin X} \boldsymbol{\mu}^*_i\le\frac{r_{l}}{\tau_{l}(\boldsymbol{\mu}^\dagger)} m\\
&\quad=\sum_{i\in A} \frac{r_{i}}{\tau_{i}(\boldsymbol{\mu}^\dagger)} \mu^\dagger_i\leq \text{SW}(\boldsymbol{\mu}^\dagger),
\end{aligned}
 \]
where $l\in A^\dagger$.  

Now we have
\[\text{SW}(\boldsymbol{\mu}^*)=\sum_{i\in X}r_{i}x_i(\boldsymbol{\mu}^*)+\sum_{i\notin X}r_{i}x_i(\boldsymbol{\mu}^*)=2{\text{SW}}(\boldsymbol{\mu}^\dagger).\]
Therefore, PoA$\le2$.

For the lower bound, we give a sequence of instances  for which the ratio inside the supremum in \eqref{eq:PoA_definition} approaches $2$.   These instances are within regimes of low discounting ($\alpha\in[0,1]$).
We consider  the constant action execution time model with $n=2$, i.e., the two-action model. Assume
 $r_1$, $r_2$, $t_2$, $m$, and $b_1$ are fixed.  Assume the resource supply rate $b_2$ is infinite, that is, $w_2$ is always equal to 0.  
Let $r_1<r_2$,
\[\frac{r_1}{t_1+w_1}=\frac{r_2}{t_2},\]
and
\[b_1(t_1+w_1)=m.\]
At the equilibrium, all agents choose action 1 and the resulting average reward is $r_1b_1$.
 The social welfare is optimized if the 
mass $b_1 w_1 = m - b_1 t_1$ of agents waiting for action $1$ take action $2$ instead.
The optimal average reward is given by
\[r_1b_1+r_2\frac{m-b_1t_1}{t_2}=r_1b_1+\frac{r_1b_1}{m}(m-b_1t_1)=2r_1b_1-\frac{b_1^2r_1t_1}{m}.\]
Thus, PoA is
\[\text{PoA}=2-\frac{b_1}{m}t
_1\rightarrow2,\]
as $t_1\rightarrow 0$.

\section{Proof of Proposition \ref{prop:stabiltity}}\label{appendix:proof_prop:stability}
Let $\mu^\dagger_1$ denote an equilibrium of the game $\barcF$. To prove $\mu^\dagger_1$ is locally asymptotically stable under dynamic \eqref{dynamic:average}, it suffices to show the following:
\begin{itemize}
    \item[(a)] When $\mu^\dagger_1\in[0,m)$, if $\mu_1>0$ and $\mu_1\rightarrow(\mu^\dagger_1)^+$, 
    \[\frac{r_1}{\tau_1(\mu_1)}-\frac{r_2}{\tau_2(\mu_1)}\le0.\]
      \item[(b)] When $\mu^\dagger_1\in(0,m]$, if $\mu_1>0$ and $\mu_1\rightarrow(\mu^\dagger_1)^-$, 
    \[\frac{r_1}{\tau_1(\mu_1)}-\frac{r_2}{\tau_2(\mu_1)}\ge0.\]
\end{itemize}
Note that in (a) and (b), it is not necessary to require the inequality to be strict since the equality holds implies that there is a continuum of equilibria near $\mu_1^\dagger$ by continuity.
We will prove (a) and (b) through a case by case analysis.

First, assume $\gamma_1\frac{\mu^\dagger_1}{t_1}+\gamma_2\frac{m-\mu^\dagger_1}{t_2}<b$, i.e., the constraint \eqref{2a1r:resource_constraint} is not tight at the equilibrium point $\mu^\dagger_1$.
Then,  
 \[\frac{r_1}{\tau_1(\mu^\dagger_1)}-\frac{r_2}{\tau_2(\mu^\dagger_1)}=\frac{r_1}{t_1}-\frac{r_2}{t_2}=\begin{cases}
     \le 0&\text{ if }\mu^\dagger_1=0\\
       =0&\text{ if }\mu^\dagger_1\in(0,m)\\
         \ge 0&\text{ if }\mu^\dagger_1=m
 \end{cases}\]

If $\mu_1>0$ and $\mu_1\rightarrow(\mu^\dagger_1)^+$ where $\mu^\dagger_1\in[0,m)$, the constraint \eqref{2a1r:resource_constraint} is not tight at $\mu_1$ by continuity. It follows that (a) holds since
 \[\frac{r_1}{\tau_1(\mu_1)}-\frac{r_2}{\tau_2(\mu_1)}=\frac{r_1}{t_1}-\frac{r_2}{t_2}=\begin{cases}
     \le 0&\text{ if }\mu^\dagger_1=0\\
       =0&\text{ if }\mu^\dagger_1\in(0,m)
 \end{cases}\]

 If $\mu_1>0$ and $\mu_1\rightarrow(\mu^\dagger_1)^-$ where $\mu^\dagger_1\in(0,m]$, the constraint \eqref{2a1r:resource_constraint} is also not tight at $\mu_1$ by continuity. It follows that  (b) holds since
 \[\frac{r_1}{\tau_1(\mu_1)}-\frac{r_2}{\tau_2(\mu_1)}=\frac{r_1}{t_1}-\frac{r_2}{t_2}=\begin{cases}
    =0&\text{ if }\mu^\dagger_1\in(0,m)\\
         \ge 0&\text{ if }\mu^\dagger_1=m
 \end{cases}\]

Second, assume $\gamma_1\frac{\mu^\dagger_1}{t_1}+\gamma_2\frac{m-\mu^\dagger_1}{t_2}=b$, i.e., the constraint \eqref{2a1r:resource_constraint} is  tight at the equilibrium point $\mu^\dagger_1$.
Then,  
 \[\frac{r_1}{\tau_1(\mu^\dagger_1)}-\frac{r_2}{\tau_2(\mu^\dagger_1)}=\frac{r_1}{t_1}-\frac{r_2}{t_2}=\begin{cases}
     \le 0&\text{ if }\mu^\dagger_1=0\\
       =0&\text{ if }\mu^\dagger_1\in(0,m)\\
         \ge 0&\text{ if }\mu^\dagger_1=m
 \end{cases}\]
Now we further assume $\frac{\gamma_1}{t_1}-\frac{\gamma_2}{t_2}>0$.

If $\mu_1>0$ and $\mu_1\rightarrow(\mu^\dagger_1)^+$ where $\mu^\dagger_1\in[0,m)$, the constraint \eqref{2a1r:resource_constraint} is still tight at $\mu_1$ since $\frac{\gamma_1}{t_1}-\frac{\gamma_2}{t_2}>0$. It follows that (a)  holds since
 \[
 \begin{aligned}
      \frac{r_1}{\tau_1(\mu_1)}-\frac{r_2}{\tau_2(\mu_1)}&=\frac{r_1}{t_1+\gamma_1w(\mu_1)}-\frac{r_2}{t_2+\gamma_2w(\mu_1)}\\
      &=\frac{r_1t_2-r_2t_1+w(\mu_1)\big(r_1\gamma_2-r_2\gamma_1\big)}{\big(t_1+\gamma_1w(\mu_1)\big)\big(t_2+\gamma_2w(\mu_1)\big)}\le0,
 \end{aligned}
 \]
where the last inequality holds because $\frac{\gamma_1}{t_1}-\frac{\gamma_2}{t_2}>0$ and $\frac{r_1}{t_1}-\frac{r_2}{t_2}\le0$ as $\mu^\dagger_1\in[0,m)$.

 If $\mu_1>0$ and $\mu_1\rightarrow(\mu^\dagger_1)^-$ where $\mu^\dagger_1\in(0,m]$, the constraint \eqref{2a1r:resource_constraint} is not tight at $\mu_1$ since $\frac{\gamma_1}{t_1}-\frac{\gamma_2}{t_2}>0$. It follows that (b)  holds since
 \[\frac{r_1}{\tau_1(\mu_1)}-\frac{r_2}{\tau_2(\mu_1)}=\frac{r_1}{t_1}-\frac{r_2}{t_2}=\begin{cases}
    =0&\text{ if }\mu^\dagger_1\in(0,m)\\
         \ge 0&\text{ if }\mu^\dagger_1=m
 \end{cases}\]

 When $\frac{\gamma_1}{t_1}-\frac{\gamma_2}{t_2}<0$,  (a) and (b)  hold by symmetry (e.g., exchange the roles of action 1 and action 2).

 When $\frac{\gamma_1}{t_1}-\frac{\gamma_2}{t_2}=0$, the equation $\gamma_1\frac{\mu_1}{t_1}+\gamma_2\frac{m-\mu_1}{t_2}=b$ always holds for $\mu_1\in[0,m]$. By a similar analysis as the first case when the constraint \eqref{2a1r:resource_constraint} is not tight, (a) and (b) still hold.

Third, assume $\gamma_1\frac{\mu^\dagger_1}{t_1}+\gamma_2\frac{m-\mu^\dagger_1}{t_2}>b$, i.e., the constraint \eqref{2a1r:resource_constraint} is   tight at the equilibrium point $\mu^\dagger_1$.
Then,  
 \[
 \begin{aligned}
      \frac{r_1}{\tau_1(\mu^\dagger_1)}-\frac{r_2}{\tau_2(\mu^\dagger_1)}&=\frac{r_1}{t_1+\gamma_1w(\mu^\dagger_1)}-\frac{r_2}{t_2+\gamma_2w(\mu^\dagger_1)}\\
      &=\begin{cases}
     \le 0&\text{ if }\mu^\dagger_1=0\\
       =0&\text{ if }\mu^\dagger_1\in(0,m)\\
         \ge 0&\text{ if }\mu^\dagger_1=m
 \end{cases}
 \end{aligned}
\]
By continuity, the constraint \eqref{2a1r:resource_constraint} is always  tight in  a small neighborhood of the equilibrium point $\mu^\dagger_1$. For $\mu_1>0$ and $\mu_1\rightarrow(\mu^\dagger_1)^+$ or $\mu_1\rightarrow(\mu^\dagger_1)^-$, we have  
\[\bar g(\mu_1):=\frac{r_1}{\tau_1(\mu_1)}-\frac{r_2}{\tau_2(\mu_1)}=\frac{r_1}{t_1+\gamma_1w(\mu_1)}-\frac{r_2}{t_2+\gamma_2w(\mu_1)},
\]
where $w(\mu_1)\ge0$ is the solution to the following equation  regarding $w$,
\[\gamma_1\frac{\mu_1}{t_1+\gamma_1w}+\gamma_2\frac{m-\mu_1}{t_2+\gamma_2w}=b.\]
By Implicit Function Theorem and some algebra, the derivative of $\bar g$ with respect to $\mu_1$ is given by,
\[
\begin{aligned}
    &\frac{\partial \bar g}{\partial \mu_1}(\mu_1)\\
    &=-\frac{\big(\gamma_2\tau_1(\mu_1)-\gamma_1\tau_2(\mu_1)\big)\big(\gamma_2\tau_1(\mu_1)r_2\tau_1(\mu_1)-\gamma_1\tau_2(\mu_1)r_1\tau_2(\mu_1)\big)}{\gamma_2^2(m-\mu_1)(\tau_1(\mu_1))^3\tau_2(\mu_1)+\gamma_2^2\mu_1\tau_1(\mu_1)(\tau_2(\mu_1))^3}.
\end{aligned}
\]

Now assume $\mu^\dagger_1=0$ and $\bar g(\mu^\dagger_1)<0$, for $\mu_1\rightarrow(\mu^\dagger_1)^+$, $\bar g(\mu^\dagger_1)<0$ by continuity. Hence, (a) is true.
Assume $\mu^\dagger_1=m$ and $\bar g(\mu^\dagger_1)>0$, for $\mu_1\rightarrow(\mu^\dagger_1)^-$, $\bar g(\mu^\dagger_1)>0$ by continuity. Hence, (b) is true.

Thus, we can assume $\bar g(\mu^\dagger_1)=0$ at the equilibrium point $\mu^\dagger_1$. Then it suffices to show that $\bar g(\mu_1)$ is decreasing in $\mu_1$ in a small neighborhood of $\mu^\dagger_1$. This is because if  $\mu_1\rightarrow(\mu^\dagger_1)^+$, $\bar g(\mu_1)\le \bar g(\mu_1^\dagger)=0$, (a) is true and if  $\mu_1\rightarrow(\mu^\dagger_1)^-$, $\bar g(\mu_1)\ge \bar g(\mu_1^\dagger)=0$, (b) is true. $\bar g(\mu_1)$ is indeed decreasing in $\mu_1$ in a small neighborhood of $\mu^\dagger_1$. Since $\bar g(\mu^\dagger_1)=0$, i.e., $r_2\tau_1(\mu_1^\dagger)=r_1\tau_2(\mu_1^\dagger)$, it follows that, 
\[\frac{\partial \bar g}{\partial \mu_1}(\mu_1^\dagger)=-\frac{r_1\tau_2(\mu_1^\dagger)\big(\gamma_2\tau_1(\mu_1^\dagger)-\gamma_1\tau_2(\mu_1^\dagger)\big)^2}{\gamma_2^2(m-\mu_1^\dagger)(\tau_1(\mu_1^\dagger))^3\tau_2(\mu_1^\dagger)+\gamma_1^2\mu_1^\dagger\tau_1(\mu_1^\dagger)(\tau_2(\mu_1^\dagger))^3}\,,\]
which is negative.
We have shown that (a) and (b) are always true. Therefore, the set of equilibria of game $\barcF$ is locally asymptotically stable.

\section{Proof of Proposition \ref{prop:unstable}}\label{appendix:proof_prop:unstable}

We will only prove for one case of \eqref{condition:unstable_1}, i.e., 
\begin{equation}
    \frac{\gamma_1}{\gamma_2}>\frac{t_1}{t_2},\ \frac{\gamma_1}{\gamma_2}<\frac{e^{\beta t_1}-1}{e^{\beta t_1}}.\label{condition:unstable_1a}
\end{equation}
The other case is similar by symmetry, e.g., we can exchange the roles of action 1 and action 2.

Fix any $m>0$ and $\mu_1^\dagger\in(0,m)$. We first show that there exist $(t_1, t_2, r_1,r_2, \gamma_1,\gamma_2, b, \beta)$ $\in\mathbb{R}_+^8$  satisfying conditions \eqref{condition:unstable_1} and \eqref{condition:unstable_2}. It is clear that there exist $(t_1, t_2,\gamma_1,\gamma_2, \beta)$ $\in\mathbb{R}_+^5$ satisfying  \eqref{condition:unstable_1}. Fix one $(t_1, t_2,\gamma_1,\gamma_2, \beta)$ $\in\mathbb{R}_+^5$ satisfying  \eqref{condition:unstable_1}, there must 
exist $b>0$ such that,
\begin{equation}
    b<\gamma_1\frac{
\mu_1^\dagger}{t_1}+\gamma_2\frac
{m-\mu_1^\dagger}{t_2}.\label{condition:unstable_2_b}
\end{equation}
Fix one such $b$, and since $m,b,t_1,t_2,\gamma_1,\gamma_2$ are fixed, $\tau_1$ and $\tau_2$ are well defined functions of $\mu_1$ by Lemma \ref{lemma:sojourn_time}. In fact, by \eqref{condition:unstable_2_b}, the resource constraint \eqref{2a1r:resource_constraint} is tight. $\tau_1(\mu_1^\dagger)$ and $\tau_2(\mu_1^\dagger)$  are given by $\tau_1(\mu_1^\dagger)=t_1+\gamma_1w(\mu_1^\dagger)$
and $\tau_2(\mu_1^\dagger)=t_2+\gamma_2w(\mu_1^\dagger)$ respectively where $w(\mu_1^\dagger)\ge0$ is the is the solution to the following equation regarding $w$,
\begin{equation}
    \gamma_1\frac{\mu_1^\dagger}{t_1+\gamma_1w}+\gamma_2\frac{m-\mu_1^\dagger}{t_2+\gamma_2w}=b.\label{eq:w}
\end{equation}

Now we can choose $(r_1,r_2)\in\mathbb{R}_+^2$ such that
 \begin{equation}
     \frac{r_1}{e^{\beta \tau_1(\mu_1^\dagger)}-1}=\frac{r_2}{e^{\beta \tau_2(\mu_1^\dagger)}-1}.\label{condition:unstable_2_r}
 \end{equation}
Hence, for any $m>0$ and $\mu_1^\dagger\in(0,m)$, there exist $(t_1, t_2, r_1,r_2, \gamma_1,\gamma_2, $
$b, \beta)$ $\in\mathbb{R}_+^8$  satisfying conditions \eqref{condition:unstable_1} and \eqref{condition:unstable_2}.

Now fix $m>0$, $\mu_1^\dagger\in(0,m)$ and $(t_1, t_2, r_1,r_2, \gamma_1,\gamma_2, b, \beta)$ $\in\mathbb{R}_+^8$ satisfying \eqref{condition:unstable_1} and \eqref{condition:unstable_2}. We will show that $\mu_1^\dagger$ is an unstable equilibrium of $\cFbeta$.
Define a function $g^\beta$ as,
\[g^\beta(\mu_1)=\frac{r_1}{e^{\beta \tau_1(\mu_1)}-1}-\frac{r_2}{e^{\beta \tau_2(\mu_1)}-1}.\]
By \eqref{condition:unstable_2_r}, $g^\beta(\mu_1^\dagger)=0$ and $\mu_1^\dagger$ is an equilibrium. To prove  $\mu_1^\dagger$ is unstable, we only need to show $g^\beta(\mu_1)$ is strictly increasing in $\mu_1$ in a small neighborhood of $\mu_1^\dagger$. 
Assume this is true. If $\mu_1>\mu^\dagger_1$ and
$\mu_1$ is close enough to $\mu^\dagger_1$, $g^\beta(\mu_1)> g^\beta(\mu_1^\dagger)=0$, the mass flow will move towards action 1 under dynamic \eqref{dynamic:discounted} and $\mu_1$ will continue to increase, which indicates that $\mu_1^\dagger$ is not stable. If $\mu_1<\mu^\dagger_1$ and $\mu_1$  
 is close enough to $\mu^\dagger_1$, $ g^\beta(\mu_1)< g^\beta(\mu_1^\dagger)=0$,  the mass flow will move towards action 2 under dynamic \eqref{dynamic:discounted} and $\mu_1$ will continue to decrease, which also indicates that $\mu_1^\dagger$ is not stable.

It suffices to prove that $\frac{\partial g^\beta}{\partial \mu_1}(\mu_1^\dagger)>0$. By applying Implicit Function Theorem to $w(\mu_1)$, which is the unique positive solution to equation \eqref{eq:w}, we have 
\[
\begin{aligned}
    \frac{\partial g^\beta}{\partial \mu_1}(\mu_1^\dagger)&=\frac{\tau_1(\mu_1^\dagger)\tau_2(\mu_1^\dagger)\big(\gamma_1\tau_2(\mu_1^\dagger)-\gamma_2\tau_1(\mu_1^\dagger)\big)}{\gamma_2^2(m-\mu_1^\dagger)(\tau_1(\mu_1^\dagger))^2+\gamma_2^2\mu_1^\dagger(\tau_2(\mu_1^\dagger))^2}\cdot\\
    &\quad\left(\frac{\gamma_2r_2\beta e^{\beta\tau_2(\mu_1^\dagger)}}{(e^{\beta\tau_2(\mu_1^\dagger)}-1)^2}-\frac{\gamma_1r_1\beta e^{\beta\tau_1(\mu_1^\dagger)}}{(e^{\beta\tau_1(\mu_1^\dagger)}-1)^2}\right)\\
&=\frac{\beta\tau_1(\mu_1^\dagger)\tau_2(\mu_1^\dagger)\big(\gamma_1\tau_2(\mu_1^\dagger)-\gamma_2\tau_1(\mu_1^\dagger)\big)}{\gamma_2^2(m-\mu_1^\dagger)(\tau_1(\mu_1^\dagger))^2+\gamma_2^2\mu_1^\dagger(\tau_2(\mu_1^\dagger))^2}\cdot\frac{r_1}{e^{\beta\tau_1(\mu_1^\dagger)}-1}\cdot\\
&\quad\left(\frac{\gamma_2 e^{\beta\tau_2(\mu_1^\dagger)}}{(e^{\beta\tau_2(\mu_1^\dagger)}-1)}-\frac{\gamma_1 e^{\beta\tau_1(\mu_1^\dagger)}}{(e^{\beta\tau_1(\mu_1^\dagger)}-1)}\right)>0\,,
\end{aligned}
\]
where the second equality follows from \eqref{condition:unstable_2_r} and the last inequality follows from \eqref{condition:unstable_1a} since
\[\gamma_1\tau_2(\mu_1^\dagger)-\gamma_2\tau_1(\mu_1^\dagger)=\gamma_1(t_2+\gamma_2w(\mu_1^\dagger))-\gamma_2(t_1+\gamma_2w(\mu_1^\dagger))=\gamma_1t_2-\gamma_2t_1>0,\]
and
\[\frac{\gamma_2 e^{\beta\tau_2(\mu_1^\dagger)}}{(e^{\beta\tau_2(\mu_1^\dagger)}-1)}-\frac{\gamma_1 e^{\beta\tau_1(\mu_1^\dagger)}}{(e^{\beta\tau_1(\mu_1^\dagger)}-1)}>\gamma_2-\gamma_1\frac{e^{\beta t_1}}{e^{\beta t_1}-1}>0.\]
Therefore, $\mu_1^\dagger$ is an unstable equilibrium of $\cFbeta$.

\end{document}